\documentclass[
nofootinbib,
aps,10pt,
 amsmath,amssymb,
 pra,longbibliography,twocolumn
]{revtex4-2}

\usepackage{graphicx}
\usepackage{bm}
\usepackage{xcolor}
\usepackage{array}
\usepackage{xfrac}
\usepackage{gensymb}
\usepackage[margin=2cm]{geometry}
\usepackage{mathtools}
\usepackage{lineno}
\usepackage[ colorlinks = true,              linkcolor = blue, urlcolor  =
blue,              citecolor = red,              anchorcolor = green,
]{hyperref}

\usepackage{newtxtext,newtxmath}

\providecommand{\U}[1]{\protect\rule{.1in}{.1in}}
\newtheorem{theorem}{Theorem}

\newtheorem{algorithm}{Algorithm}

\newtheorem{lemma}[theorem]{Lemma}

\newenvironment{proof}[1][Proof]{\noindent\textbf{#1.} }{\ \rule{0.5em}{0.5em}}
\numberwithin{theorem}{section}
\numberwithin{proposition}{section}
\numberwithin{definition}{section}
\numberwithin{example}{section}

\newcommand{\bra}[1]{\langle#1|}
\newcommand{\ket}[1]{|#1\rangle}
\newcommand{\innerprod}[2]{\langle #1 \vert #2 \rangle}
\newcommand{\outerproj}[1]{\vert #1 \rangle\!\langle #1 \vert}

\newcommand{\Tr}[1]{\operatorname{Tr}[#1]}

\newcolumntype{P}[1]{>{\centering\arraybackslash}p{#1}}

\allowdisplaybreaks

\begin{document}

\title{Quantum Algorithms for Realizing Symmetric, Asymmetric, and Antisymmetric Projectors}

\author{Margarite L. LaBorde}
\affiliation{Advanced Acoustics \& Seabed Warfare Branch, Naval Surface Warfare Center, Panama City, Florida 32407, USA}
\author{Soorya Rethinasamy}
\affiliation{School of Applied and Engineering Physics, Cornell University, Ithaca, New
York 14850, USA}
\author{Mark M. Wilde}
\affiliation{School of Electrical and Computer Engineering, Cornell University, Ithaca, New
York 14850, USA}

\begin{abstract}

    In quantum computing, knowing the symmetries a given system or state obeys or disobeys is often useful. For example, Hamiltonian symmetries may limit allowed state transitions or simplify learning parameters in  machine learning applications, and certain asymmetric quantum states are known to be resourceful in various applications. Symmetry testing algorithms provide a means to identify and quantify these properties with respect to a representation of a group. In this paper, we present a collection of quantum algorithms that realize projections onto the symmetric subspace, as well as the asymmetric subspace, of quantum systems. We describe how this can be modified to realize an antisymmetric projection as well, and we show how projectors can be combined in a systematic way to effectively measure various projections in a single quantum circuit. Using these constructions, we demonstrate applications such as testing for Werner-state symmetry and estimating Schmidt ranks of bipartite states, supported by experimental data from IBM Quantum systems. This work underscores the pivotal role of symmetry in simplifying quantum calculations and advancing quantum information tasks.
\end{abstract}

\maketitle

\tableofcontents

\section{Introduction}

Symmetry has been prevalent throughout the field of physics since the introduction of Noether's theorem~\cite{Noether1918}, and perhaps even earlier through its prominence in mathematics and observation in crystalline materials~\cite{fedorov1891symmetry}. A  recent direction of scholarship focuses on combining symmetry with quantum information theory, demonstrating how the mathematical beauty of symmetry can give rise to physically interesting phenomena in this regime. For instance, the concept of Noether's theorem, an integral link between physical symmetries and conserved quantities, was expanded upon in~\cite{marvian2014extending}, which investigated the resourcefulness of quantum states under free channels. Similarly, the resource theory of asymmetry~\cite{Marvian2013asymmetry} illustrates a connection between quantum computational resources and symmetric properties of said quantum resources. Finally, Schur--Weyl duality and its links to permutation symmetry enable quantum de Finetti theorems~\cite{caves2002unknown, konig2005finetti,christandl2007one}, which in turn provide the backbone for many  tests of entanglement, particularly in the case of $k$-extendibility~\cite{W89,DPS02,DPS04}. 

In quantum computing, symmetry can in principle be used to simplify calculations, speed up computational tasks, or tackle classically intractable problems. 
For instance, fully permutationally invariant Hamiltonians have been shown to be efficiently simulatable on classical hardware~\cite{anschuetz2023efficient}. Quantum algorithms for testing symmetry of states~\cite{laborde2021testing}, Hamiltonians~\cite{laborde2022quantum}, and open quantum systems~\cite{bandyopadhyay2023efficient} have been investigated in recent work, and some were shown to delineate a hierarchy of quantum complexity classes~\cite{rethinasamy2023quantum}. Furthermore, incorporating symmetry into quantum machine learning can reduce parameter search space and speed up calculations~\cite{ragone2022representation, larocca2022group,meyer2023exploit}. 

In physics, a natural complement to symmetry is antisymmetry. Antisymmetry is of course one particular type of symmetry, observable in nature in fermionic systems. Antisymmetric properties are pertinent for a variety of physical phenomena. Time reversal and spatial inversion antisymmetries occur in magnetic crystallography, where some can be used to classify physical properties (see~\cite{Padmanabhan2020antisymmetry} and references therein). Additionally, the antisymmetry inherent in the quantum marginal problem for fermionic systems~\cite{schilling2015quantum} implies generalized Pauli constraints for these systems. Particularly relevant to quantum information tasks, antisymmetric projections have been employed in relation to the Schmidt rank~\cite{JLV22} of quantum states---indeed, we will use this result later in this work.   

In this paper, we propose various quantum algorithms for testing symmetry, antisymmetry, and the asymmetry of states. In more detail, we delineate algorithms for projecting quantum states onto the symmetric and antisymmetric subspaces of the standard representation of a permutation group and then generalize this procedure to implement projectors in sequence, a process we term projector concatenation. In combining these mathematical projectors, we demonstrate how our tests enable us to examine relevant quantum information tasks, such as testing for Werner-state symmetry and estimating the Schmidt rank of a pure bipartite state. For these applications, we include example data from experiments run on the IBM Quantum ibm\_kyoto and ibm\_osaka computers. We also show how our algorithms generalize to more generic projectors. Finally, we examine how a similar framework can realize a quantum switch that estimates commutators of observables.

 The structure of our paper is as follows. In Section~\ref{sec:background}, we review previously established formalism for symmetry tests and introduce the symmetric and antisymmetric projectors. In Section~\ref{sec:concat}, we show how additional uncomputing steps allow for realizing both projectors in a single quantum circuit. We then show how this method can be conducted in sequence to test a state's projection onto various subspaces by means of a single circuit. We give examples for symmetric subspaces and realize the difference of two or more projectors. In Section~\ref{sec:applications}, we discuss applications of these techniques. In particular, we show how an antisymmetry test can be used to estimate Schmidt rank and how the difference of two projectors can be used as a test for Werner-state symmetry. We also show how an antisymmetrized circuit can be used to estimate commutators. Finally, we conclude in Section~\ref{sec:discussion} with a summary of our results and comments on remaining open problems.

\section{Symmetric and Antisymmetric Projections}\label{sec:background}

\subsection{Review of Symmetric Projections}

This work naturally rests on a foundation of symmetry tests; we will begin there and then build upon it recursively, generalizing more in each step. Let us begin by defining precisely what we mean by symmetry and, in turn, antisymmetry. 

Consider a finite group $G$ with a unitary representation $U:G\to \mathcal{U}(V)$ on a vector space $V$. For such a representation~$U$ of a group $G$, we define the $G$-symmetric subspace by
\begin{equation}
    \mathcal{S}^G \coloneqq \{\ket{\psi}\in V\ \vert\ U(g)\ket{\psi}=\ket{\psi}, \ \forall g\in G\}.
    \label{eq:symm-def}
\end{equation}
Vectors that belong to $\mathcal{S}^G$ are called $G$-symmetric vectors. If $V=\mathcal{H}$ for some Hilbert space $\mathcal{H}$, then the set of $G$-symmetric density matrices is given by
\begin{equation}
    \mathcal{S}^G_\rho \coloneqq \left \{\rho \in \mathcal{D}(\mathcal{H}) \ \vert\ U(g) \rho U^\dag (g),\ \forall g\in G\right \},
\end{equation}
where $\mathcal{D}(\mathcal{H})$ denotes the set of density operators acting on a Hilbert space $\mathcal{H}$.

A stronger notion of symmetry for density operators, called $G$-Bose symmetry, is discussed at length in~\cite{laborde2021testing}, wherein a mixed state $\rho$ is said to be $G$-Bose symmetric if 
\begin{equation}
\label{eq:G-symmetry-condition}
    \rho = \Pi_G\rho\Pi_G,
\end{equation}
with the projector onto the $G$-Bose symmetric subspace defined as
\begin{equation}\label{eq:groupprojection}
    \Pi_{G} \coloneqq \frac{1}{|G|} \sum_{g \in G} U(g)\, .
\end{equation}
The projection onto the symmetric subspace has been well studied and is of interest for many applications (see, e.g.,~\cite{marvian2012symmetry,MS14,harrow2013church}). It should be noted for clarity that in this work we use the term ``projector" to refer to the mathematical object above and ``projection" to refer to the image of a quantum state after being acted on by such a projector.

It was shown that a test for $G$-Bose symmetry could be achieved using a generalized phase estimation algorithm, as proposed in \cite[Chapter~8]{harrow2005applications}, and Figure~\ref{fig:GBStest} depicts this algorithm. Since $G$-Bose symmetry is a strictly stronger notion than $G$-symmetry, any state exhibiting the former symmetry naturally demonstrates the latter symmetry as well. 

\begin{figure}
\begin{center}
\includegraphics[width=0.8\columnwidth]{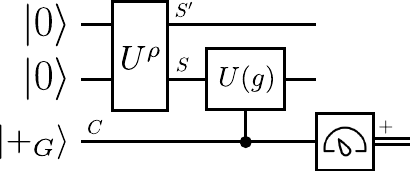}
\end{center}
\caption{Figure demonstrating the generalized phase estimation applied to test for $G$-Bose-symmetry of a state, as in~\eqref{eq:G-symmetry-condition}. The unitary $U^{\rho}$ prepares a purification $\psi_{S^{\prime}S}$ of the state $\rho_{S}$. The final measurement box with the plus-sign to the right of it indicates to accept only if the initial superposition state $\ket{+_G}$ is measured.  Its acceptance probability is equal to $\operatorname{Tr}[\Pi_{G} \rho_{S} \Pi_G] = \operatorname{Tr}[\Pi_{G} \rho_{S}]$.}
\label{fig:GBStest}
\end{figure}

For the case where $G$ is the symmetric group $S_k$ on $k$ letters, there is a natural representation $U(\sigma)$ whose action on the tensor-power Hilbert space $\mathcal{H}^{\otimes k}$ is given by permuting the Hilbert spaces according to the permutation elements $\sigma$ of~$S_k$. The space of states that is invariant under this action is known as the symmetric subspace:
\begin{equation}
    \mathcal{S}^{S_k}\coloneqq \{\ket{\psi}\in\mathcal{H}^{\otimes k}\ \vert\ U(\sigma)\ket{\psi}=\ket{\psi}, \ \forall\sigma\in S_k\}.
\end{equation}
The projection onto $\mathcal{S}^{S_k}$  is given by acting on a state with the projector 
\begin{equation}
\label{eq:sk_projector}
    \Pi_{S_k} \coloneqq \frac{1}{k!} \sum_{\sigma \in S_k} U(\sigma) \, ,
\end{equation}
where, unless stated otherwise, the representation $\{U(\sigma)\}_{\sigma\in S_k }$ of the symmetric group is generally taken to be the standard representation. Explicitly, we consider permutations $\sigma\in S_k$ of the Hilbert space in such a manner that the action of $U(\sigma)$ on a state $\ket{\phi}_{A_1A_2 \cdots A_k}$ is realized as
\begin{equation}\label{eq:permutationaction}
    U(\sigma)  \ket{\phi}_{A_1A_2 \cdots A_k} = \ket{\phi}_{A_{\sigma(1)}A_{\sigma(2)} \cdots A_{\sigma(k)}}\, ,
\end{equation}
where $\{U(\sigma)\}_{\sigma \in S_k}$ is a unitary representation of $S_k$ in the same manner as those given in~\cite{barenco1997stabilization}. Equivalently,
\begin{equation}
    U(\sigma) = \sum_{i_1, \ldots, i_k\in[d]} |i_{\sigma^{-1}(1)}, \ldots, i_{\sigma^{-1}(n)}\rangle \! \langle i_1, \ldots, i_k|,
\end{equation}
as defined in \cite[page~3]{harrow2013church}. 

Given the projector onto the symmetric subspace, Ref.~\cite{barenco1997stabilization} gave a construction to test whether an input state $\rho$ obeys the $S_k$-Bose symmetry condition
\begin{equation} 
\label{eq:sk-symmetry-condition}
    \rho = \Pi_{S_k} \rho \Pi_{S_k}\, .
\end{equation}
Using the fact that the symmetric group $S_k$ can be generated by acting on the group $S_{k-1}$ with all transpositions of the type $(i\ k)$, the entire symmetry test can be achieved using only transpositions and particular choices of ancilla qubits. 

We will now briefly review the construction, but we urge readers to review the more detailed explanations given in~\cite{barenco1997stabilization} and~\cite{bradshaw2022cycle}. The circuit for $S_k$ is generated recursively given a circuit for $S_{k-1}$, and $S_2$ can be constructed explicitly. Define a control state
\begin{equation}
    \ket{+_{S_k}}_C \coloneqq \bigotimes_{j=1}^{k} A_j \ket{0}^{\otimes j-1} \, ,
\end{equation}
where $A_j $ is a unitary with the following action:
\begin{align}
\label{eq:defAj}
    A_j \ket{0}^{\otimes j-1} &= \frac{1}{\sqrt{j}}\ket{0}^{\otimes j-1} + \sqrt{\frac{j-1}{j}} \ket{W_{j-1}},  \\
    &= \frac{1}{\sqrt{j}} \left( \ket{0}^{\otimes j-1} + \sum_{i=1}^{j-1}\ket{2^{i-1}} \right),
\end{align}
and $\ket{2^{i-1}}$ is the computational basis state with a one in the $i$-th position and zeros everywhere else. The $j$-th iteration controls an array of $j-1$ C-SWAP gates corresponding to transpositions of the form $( i\ j )$, where $i$ runs from $1$ to $j-1$. 

To see this, consider generating the test for $S_3$. Begin with the circuit for $S_2$, the base case, which is equivalent to the SWAP test between Hilbert spaces $\mathcal{H}_1$ and $\mathcal{H}_2$. The control state is then given by one qubit in the state
\begin{equation}
    \ket{+_{S_2}}=\frac{1}{\sqrt{2}}\left (\ket{0}+\ket{1} \right ).
\end{equation}
To create the test for $S_3$, we append to this a layer where $j-1 = 2$. Thus two ancilla qubits are necessary to control two new C-SWAP gates, and they are initialized in the state 
\begin{equation}
\frac{1}{\sqrt{3}}\left (\ket{00} + \ket{01}+\ket{10} \right),
\end{equation}
for an overall control state of
\begin{equation}\label{eq:s3examplecontrol}
    \ket{+_{S_3}}\coloneqq \frac{1}{\sqrt{3}}\left (\ket{00} + \ket{01}+\ket{10}\right ) \otimes \ket{+_{S_2}}\, .
\end{equation}
Controlling from~\eqref{eq:s3examplecontrol}, if the additional two ancilla qubits are in the state $\ket{01}$, the second and third Hilbert spaces are swapped, realizing the transposition $( 2\ 3)$. If two new ancilla qubits are in the state $\ket{10}$, the first and third Hilbert spaces are swapped, realizing $( 1\ 3)$. As before, $\ket{1}$ on the original ancilla realizes a swap between the first and second Hilbert spaces, or the transposition $( 1\ 2)$. There is some freedom in choosing which transposition corresponds to which control states, but this general procedure reflects the construction given in Figure~\ref{fig:generalpermutation} and is discussed at length in~\cite{bradshaw2022cycle}. With this recipe, the projector $\Pi_{S_k}$ can be realized for every $k$ using circuit depth $\mathcal{O}(k^2)$.

\begin{figure}
\begin{center}
\includegraphics[width=\columnwidth]{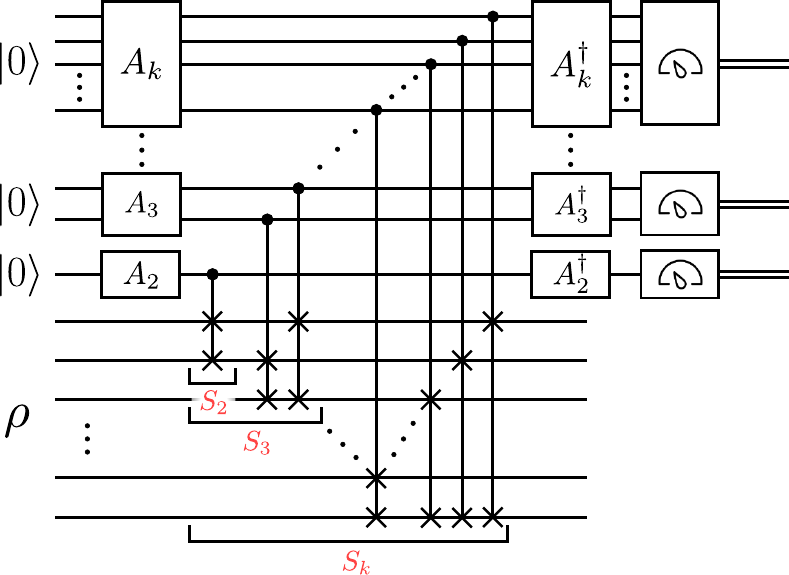}
\end{center}
\caption{Quantum circuit for systematically realizing a test for the symmetric group, using the procedure described in~\cite{barenco1997stabilization}. The definition of the unitary $A_j$ is given in~\eqref{eq:defAj}. The symbol $\rho$ refers to a state of $k$ qubits or systems.}
\label{fig:generalpermutation}
\end{figure}

\subsection{Antisymmetric Projection}

\label{sec:antisymmetry}

Let us now consider antisymmetry in terms of the standard representation of a group---or rather, in terms of permutations. Cayley's theorem~\cite{Dummit_Foote} states that any finite group $G$ is isomorphic to a subgroup of a permutation group, and so we can safely restrict ourselves to considering group elements as being isomorphic to permutations without loss of generality. An element $\ket{\psi}$ in the appropriate vector space $V$ is antisymmetric with respect to a representation $U$ of a group $G$ if, for every transposition $\sigma^{\prime} \in G$, we have that 
\begin{equation}\label{eq:transpositionaction}
U(\sigma^{\prime})\ket{\psi} = (-1) \ket{\psi} \, .
\end{equation}
Equivalently, an antisymmetric vector $\ket{\psi}$ is defined by
\begin{equation}
    U(\sigma) \ket{\psi} = \text{sgn}(\sigma) \ket{\psi}\, ,
\end{equation}
where $\text{sgn}(\sigma)$ denotes the sign of the permutation $\sigma$.

To model this action on quantum states, let us consider permutations of systems acting as in~\eqref{eq:permutationaction}, and define single transpositions $\sigma_{ij}$ that swap the $i$-th and $j$-th systems, and some multipartite state $\rho_{A_1 \cdots A_k}$,
which is defined on a tensor-product Hilbert space $\mathcal{H}^{\otimes k}$. Then we propose a  test analogous to~\eqref{eq:sk-symmetry-condition}, wherein we define the antisymmetric projector
\begin{equation}\label{eq:antisymmetricprojector}
    \Pi_{\textrm{anti}_k} \coloneqq \frac{1}{k!} \sum_{\sigma \in S_k} \operatorname{sgn}(\sigma) U(\sigma) \, ,
\end{equation}
where $\{U(\sigma)\}_{\sigma \in S_k}$ is a unitary representation of $S_k$ in the same manner as that given in~\cite{barenco1997stabilization}; specifically, we choose a unitary representation such that it obeys the desired group action described in~\eqref{eq:permutationaction}. 

Given this definition of antisymmetry, a corresponding algorithmic test can be distilled from that  in~\cite{barenco1997stabilization, bradshaw2022cycle}. The circuit construction is depicted in Figure~\ref{fig:antiSymTest}. 

\begin{figure}
\begin{center}
\includegraphics[width=\columnwidth]{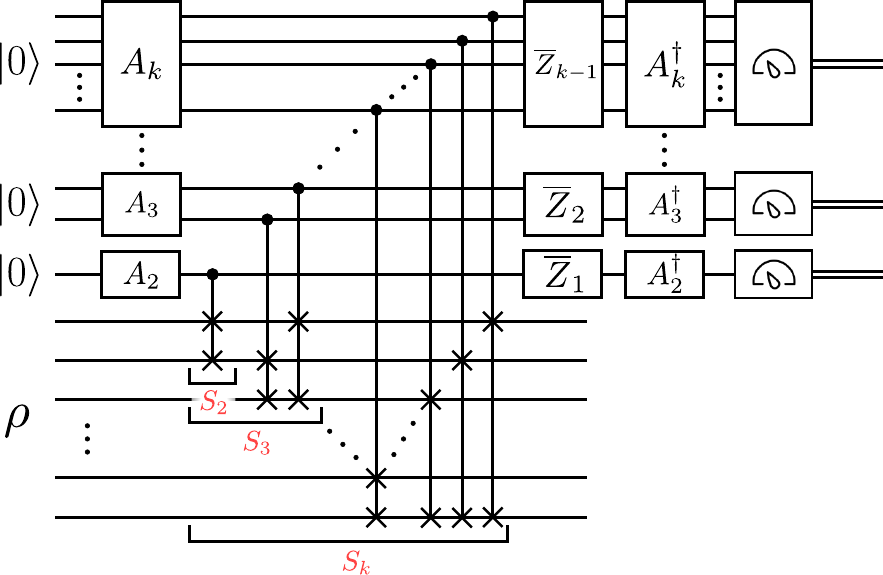}
\end{center}
\caption{Figure demonstrating how to systematically generate a test for antisymmetry. The definition of the unitaries $A_j$ is given in~\eqref{eq:defAj}, and $\overline{Z}_n$ is defined in~\eqref{eq:Zbar}. The symbol $\rho$ refers to a state on $k$ qubits or systems.
}
\label{fig:antiSymTest}
\end{figure}

As a reminder, the control state for that algorithm is 
\begin{equation}
    \ket{+_{S_k}}_C \coloneqq \bigotimes_{j=1}^{k} A_j \ket{0}^{\otimes j-1} \, ,
    \label{eq:control-state-anti-sym}
\end{equation}
where the action of the unitary $A_j$ is given in~\eqref{eq:defAj}. Here we take the same approach shown in Figure~\ref{fig:generalpermutation}, but additionally perform the operation $\overline{Z}_{k-1} \otimes \cdots \otimes \overline{Z}_1$ after the C-SWAP gates. 
The operation $\overline{Z}_n$ is defined as follows:
\begin{equation}
    \label{eq:Zbar}
    \overline{Z}_n \coloneqq Z^{\otimes n},
\end{equation}
where $n$ is the number of qubits in the register. After this change, the effective controlled unitary in each instance is given by 
\begin{multline}
    \operatorname*{C-antiSWAP}({i,j}) \coloneqq \ket{0}\!\bra{0}_C \otimes \mathbb{I}_{i,j} \\
    + \ket{1}\!\bra{1}_C \otimes (-1)U(\sigma_{i,j}),
\end{multline}
where $C$ is the specific qubit controlling the swap and $i,j$ are the spaces to be swapped. Additionally, $U(\sigma_{i,j})$ is a SWAP gate acting on the $i$th and $j$th registers. Replacing all of the controlled-SWAPs in the algorithm with $\operatorname*{C-antiSWAP}_{i,j}$ gives an overall implementation of the antisymmetric projector in~\eqref{eq:antisymmetricprojector}:
\begin{equation}
    \bra{+_{S_k}}_C \ \prod_{j=2}^{k} \prod_{i=1}^{j-1} \operatorname*{C-antiSWAP}(i,j) \ket{+_{S_k}}_C = \Pi_{\textrm{anti}_k} \, .
\end{equation}
Therefore, a test of antisymmetry can be easily generated from a corresponding test of symmetry.

\section{Concatenating Projectors}\label{sec:concat}

\subsection{Uncomputing to Expand upon a Symmetry Test}\label{sec:uncomputing}
The approach in Figure~\ref{fig:GBStest} can be expanded upon by means of uncomputation. Uncomputation refers to a technique commonly employed in quantum algorithms wherein a computation is performed in reverse after the desired output has been coherently stored~\cite{Bennett1973logical,aaronson2015classification}. Of course, as with many quantum algorithms, the control state is already undergoing some uncomputation. However, by reversing the action on the remaining systems, the outcome can be used to measure not only the projection onto the symmetric subspace but also the projection onto the asymmetric space as well.

The rough sketch of the algorithm goes as thus: perform a symmetry testing circuit on the state in question. Then, instead of performing the usual measurement, store the acceptance outcome in an additional ancilla by performing a CNOT controlled off of the all-zeros outcome. After this, uncompute the symmetry test by reversing the order of the gates in the original test and taking their conjugate transposes. If the new ancilla qubit is measured, the probability that we measure one is $\left\Vert \Pi_G \ket{\psi} \right\Vert_2^2$ and the probability to measure zero is $\left\Vert \left( \mathbb{I} - \Pi_G\right) \ket{\psi} \right\Vert_2^2$. Furthermore, the post-measurement state will be proportional to $\Pi_G \ket{\psi}$ in the first case and to $\left( \mathbb{I} - \Pi_G\right) \ket{\psi}$ in the second case.  Therefore, this algorithm can measure the symmetric and asymmetric projections of a state using a single circuit and in a coherent manner.

We will now give a more thorough, mathematical description of the algorithm. 

\begin{algorithm}
Given a unitary representation $\left\{  U(g)\right\} _{g\in G}$\ of a group $G$, we can implement the group projection from~\eqref{eq:groupprojection} with the correct output probability distribution and post-measurement states, by performing the following steps:
\begin{enumerate}
\item Prepare the state $\ket{+_G}_{C}\coloneqq \frac{1}{\sqrt{\left\vert G\right\vert}}\sum_{g\in G}\ket{g}_{C}$.
\item Perform the controlled unitary $\sum_{g\in G}\ket{g}\!\bra{g}_{C}\otimes U(g)$ from the control register to the data register.
\item Prepare a control qubit in the state $\ket{0}_{C^{\prime}}$ and perform the controlled unitary 
\begin{equation}
X_{C^{\prime}} \otimes \ket{+_G}\!\bra{+_G}_{C} + \mathbb{I}_{C^{\prime}} \otimes \left(  \mathbb{I}-\ket{+_G}\!\bra{+_G}\right)_{C}.
\end{equation}
\item Perform the controlled unitary $\sum_{g\in G}\ket{g}\!\bra{g}_{C}\otimes U^{\dag}(g)$ from the first control register $C$ to the data register.
\item Optional:\  discard the control register~$C$.
\item Measure the register $C^{\prime}$ in the computational basis.
\end{enumerate}
\end{algorithm}

\begin{figure}
\begin{center}
\includegraphics[width=\columnwidth]{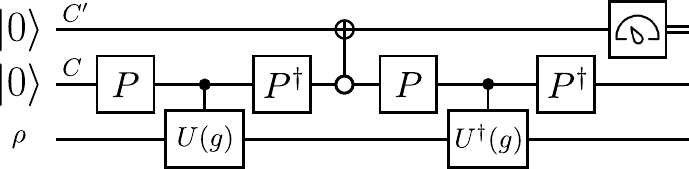}
\end{center}
\caption{Figure demonstrating the use of uncomputation to generalize the previous symmetry test. The gate $P$ acts on the all-zeros state to create a superposition $\ket{+_G}$, and $\{U(g)\}_{g\in G}$ denotes the unitary representation of $G$. The target qubit is measured to be in the $\ket{0}$ state with probability $\operatorname{Tr}[\Pi_G\rho]$ or in the $\ket{1}$ state with probability $\operatorname{Tr}[(\mathbb{I}-\Pi_G)\rho]$. Furthemore, the post-measurement state is proportional to $\Pi_G\rho\Pi_G$ or $(\mathbb{I}-\Pi_G)\rho(\mathbb{I}-\Pi_G)$, depending on the measurement outcome.}
\label{fig:uncomputing}
\end{figure}

Figure~\ref{fig:uncomputing} depicts the approach. Following the procedure above, the initial state of the algorithm after step one is $    \ket{0}_{C'}\ket{+_G}_C\ket{\psi}$ where
\begin{equation}\label{eq:plusstate}
    \ket{+_G}_C \coloneqq \frac{1}{\sqrt{|G|}} \sum_{g \in G} \ket{g}_{C} \, ,
\end{equation}
a choice of control state we will employ often. After step two, the overall state of the system is
\begin{equation}
     \frac{1}{\sqrt{|G|}} \sum_{g \in G} \ket{0}_{C'}\ket{g}_{C} U(g) \ket{\psi} \, .
\end{equation}

\begin{figure}
\begin{center}
\includegraphics[width=\columnwidth]{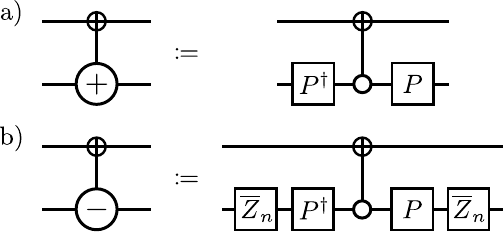}
\end{center}
\caption{For convenience, we define a plus or minus controlled gate. The gate $P$ prepares the state $\ket{+_G}$ on $n$ qubits; by convention, we say that applying a $Z$ gate to each qubit of the $\ket{+_G}$ state generates the $\ket{-_G}$ state. The controlled gate flips the target qubit if a) the control is in the $\ket{+_G}$ state or b) the control is in the $\ket{-_G}$ state, respectively. The operation $\overline{Z}_n$ is defined in~\eqref{eq:Zbar}. }
\label{fig:plusandminus}
\end{figure}

\noindent
Here, we depart from the procedure followed in~\cite{laborde2021testing}. Instead of measuring the control qubit after uncomputing, perform a controlled operation using the $\ket{+_G}$ as a control state instead of a $\ket{0}$ or $\ket{1}$ instance. This change of control is shown in Figure~\ref{fig:plusandminus}a), where we define a circuit diagram for use in future algorithms. 
After the controlled unitary in step three, the state of the system is
\begin{multline}\label{eq:uncomputestep3part1}
        \left[X_{C^{\prime}} \otimes \ket{+_G}\!\bra{+_G}_{C} + \mathbb{I}_{C^{\prime}} \otimes \left(\mathbb{I} -\ket{+_G}\!\bra{+_G}\right)_{C}\right]\  \times \\
        \frac{1}{\sqrt{|G|}}\sum_{g \in G} \ket{0}_{C'} \ket{g}_{C} U(g) \ket{\psi} \,.
\end{multline}

Before continuing, let us observe a few facts. First, building from the logic of the symmetry test depicted in Figure~\ref{fig:GBStest}, we know that
\begin{equation}\label{eq:projectorassertion}
    (\bra{+_G}_C \otimes \mathbb{I})\frac{1}{\sqrt{|G|}} \sum_{g \in G} \ket{g}_C \otimes U(g) = \Pi_G \, .
\end{equation}
Indeed, the algorithm presented in~\cite{laborde2021testing} depends on this fact. Furthermore, since $\Pi_G$ is a projector formed from the sum over the unitary representation, observe that 
\begin{equation}\label{eq:projectorfact}
    U(g) \Pi_G = U^\dag(g) \Pi_G = \Pi_G \, .
\end{equation}

Equipped with these facts, we can rewrite~\eqref{eq:uncomputestep3part1} as
\begin{multline}
    \ket{1}_{C'}\ket{+_G}_C\Pi_G \ket{\psi} - \ket{0}_{C'}\ket{+_G}_C\Pi_G\ket{\psi} \\
    + \frac{1}{\sqrt{|G|}}\sum_{g \in G} \ket{0}_{C'}\ket{g}_{C} U(g) \ket{\psi}\,. 
\end{multline}
From here, we progress to step four. Using the property in~\eqref{eq:projectorfact}, the state of the system after step four is
\begin{align}
    & \ket{1}_{C'}\ket{+_G}_C\Pi_G \ket{\psi} - \ket{0}_{C'}\ket{+_G}_C\Pi_G\ket{\psi} + \ket{0}_{C'}\ket{+_G}_C \ket{\psi} \notag \\
    & = \ket{1}_{C'}\ket{+_G}_C (\Pi_G \ket{\psi}) + \ket{0}_{C'}\ket{+_G}_C (\mathbb{I} - \Pi_G)\ket{\psi} \\
    & = \ket{+_G}_C \otimes \big[\ket{1}_{C'} (\Pi_G \ket{\psi}) + \ket{0}_{C'} (\mathbb{I} - \Pi_G)\ket{\psi} \big]\, .
\end{align}
Since the $C$ register is in tensor product with the other registers, it can be discarded. We will elect to do so, in order to clean up notation. This gives the final state of the system:
\begin{equation}\label{eq:uncomputefinal}
    \ket{1}_{C'}\Pi_G \ket{\psi} + \ket{0}_{C'}\left( \mathbb{I} - \Pi_G\right) \ket{\psi} \, .
\end{equation}
From~\eqref{eq:uncomputefinal}, we can clearly see that, if we measure register~$C'$ in the standard basis, then we will detect $C'$ to be in the state~$\ket{1}$ with probability
\begin{equation}\label{eq:uncomputesym}
    \Pr(\ket{1}_{C'}) = \left\Vert \Pi_G \ket{\psi} \right\Vert_2^2 \, ,
\end{equation}
with post-measurement state
\begin{equation}
    \frac{\Pi_G \ket{\psi}}{\left\Vert \Pi_G\ket{\psi} \right\Vert_2^2}\, .
\end{equation}
Furthermore, we will detect the state~$\ket{0}$ with probability
\begin{equation}\label{eq:asymmetric1}
    \Pr(\ket{0}_{C'}) = \left\Vert \left( \mathbb{I} - \Pi_G\right) \ket{\psi} \right\Vert_2^2 \, ,
\end{equation}
with post-measurement state
\begin{equation}
    \frac{\left( \mathbb{I} - \Pi_G\right) \ket{\psi}}{\left\Vert \left( \mathbb{I} - \Pi_G\right) \ket{\psi} \right\Vert_2^2}\, .
    \label{eq:post-meas-state-asym}
\end{equation}

The result in~\eqref{eq:asymmetric1}--\eqref{eq:post-meas-state-asym} can be understood as realizing the projection onto the asymmetric subspace. Thus, we have accomplished our task of generalizing the symmetric test to realize both projections simply by performing additional uncomputation. Previous implementations \cite{laborde2021testing} of this type of symmetry test gave an estimate of the projection onto the symmetric subspace by measuring the ancilla qubit. By adding additional ancilla qubits to the control register and then acting on the data register with the Hermitian conjugate of the symmetry testing algorithms, we have created a way to realize either the symmetric and asymmetric projections probabilistically. Furthermore, we know that the post-measurement state will either be  symmetric or asymmetric according to the outcome of the measurement.

\subsection{Realizing the Antisymmetric Projection}

A similar method can be employed for the antisymmetric test given in Section~\ref{sec:antisymmetry},  which gives analogous results. Define the antisymmetric projector for a group $G$ as
\begin{equation}
    \Pi_{\mathcal{A}} \coloneqq \frac{1}{|G|}\sum_{g\in G} \operatorname{sgn}(g) U(g) .
\end{equation}
This can be understood as a specific symmetry where the representation is defined as $\{\operatorname{sgn}(g) U(g)\}_{g\in G}$, that is, a representation identical to that of the symmetric group up to the sign of the group element. While the sign of the group element may seem nebulously defined, for finite groups there is a clear meaning, as previously stated. Cayley's theorem states that all groups are isomorphic to a subgroup of the full permutation group. The elements of a permutation group have a well-defined sign, as even permutations have sign $+1$ and odd-order permutations have sign $-1$. Thus, to determine the sign of a group element, it suffices to determine the sign of its corresponding permutation under Cayley's theorem.

For completeness, we now detail the procedure.

\begin{algorithm}
    The algorithm consists of the following steps:

\begin{enumerate}

\item Prepare the state $\ket{+_G}_{C}\coloneqq \frac{1}{\sqrt{\left\vert G\right\vert}}\sum_{g\in G}\ket{g}_{C}$.

\item Perform the controlled unitary $\sum_{g\in G}\ket{g}\!\bra{g}_{C}\otimes U(g)$ from the control register to the data register. 

\item Act on every control qubit, changing the sign of all odd-parity basis states to $-1$. This is equivalent to applying the unitary $\sum_{g \in G} \operatorname{sgn}(g)\ket{g}\!\bra{g}_{C}$.

\item Prepare another control qubit in the state 
$\ket{0}_{C^{\prime}}$ and perform the controlled unitary
\begin{equation} 
    X_{C^{\prime}} \otimes \ket{+_G}\!\bra{+_G}_{C} + \mathbb{I}_{C^{\prime}} \otimes \left(  \mathbb{I}-\ket{+_G}\!\bra{+_G}\right)_{C}.
\end{equation}

\item Act on the control qubits once more to change the sign of all odd-parity basis states (equivalent to applying the unitary $\sum_{g \in G} \operatorname{sgn}(g)\ket{g}\!\bra{g}_{C}$).

\item Perform the controlled unitary $\sum_{g\in G}\ket{g}\!\bra{g}_{C}\otimes U^{\dag}(g)$ from the control register to the data register.

\item Optional:\ 
discard the control register~$C$.

\item Measure the register $C^{\prime}$ in the computational basis.
\end{enumerate}
\end{algorithm}

The first two steps are identical to those used for the symmetric projection, and the state of the system afterward is given by
\begin{equation}
     \frac{1}{\sqrt{|G|}} \sum_{g \in G} \ket{0}_{C'}\ket{g}_{C} U(g) \ket{\psi} \, ,
\end{equation}
as before. Now, before moving on to the third step, define the state 
\begin{equation}\label{eq:minusstate}
    \ket{-_G} \coloneqq \frac{1}{\sqrt{|G|}} \sum_{g \in G} \operatorname{sgn}(g) \ket{g} \, .
\end{equation}
It is easily checked that, for permutation groups, $\overline{Z}_n \ket{+_G} =  \ket{-_G}$ when the control state on $n$ qubits is constructed in such a way that odd-parity basis states in the control are used to represent odd-parity permutations. As alluded to in Section~\ref{sec:antisymmetry}, any finite group can be thought of as a subgroup of a permutation group, and so this manner of creating $\ket{-_G}$ from $\ket{+_G}$ generalizes in principle. Indeed, this is a natural way of constructing the control register. If the identity group element $g = e$ is associated to $\ket{0}^{\otimes n}$ for $n$ ancilla qubits, then the first transposition $g = (1\ 2)$ in $S_2$ can be associated to $\ket{0}^{\otimes n-1} \otimes \ket{1}$, and so on for each additional transposition. This exactly follows the construction in Figure~\ref{fig:generalpermutation}.
Thus we can combine steps 3-5 into a single step where the effective controlled unitary is instead
\begin{equation}
X_{C^{\prime}} \otimes \ket{-_G}\!\bra{-_G}_{C} + \mathbb{I}_{C^{\prime}} \otimes \left(  \mathbb{I}-\ket{-_G}\!\bra{-_G}\right)_{C}.
\end{equation}
Note that this is completely equivalent to the $\operatorname*{C-antiSWAP}$ construction used previously in Section~\ref{sec:antisymmetry} for the antisymmetry test; the Pauli-$Z$ gates have merely been absorbed into the control state instead of the controlled unitary. Similar to how the $\ket{+_G}$ state was used to alter a controlled-NOT gate, this state will also be used in that manner. This is depicted in Figure~\ref{fig:plusandminus}b). 

After applying the new controlled unitary as described above, the state of the system is given by
\begin{multline}\label{eq:uncomputeAntiStep5}
    \left[ X_{C^{\prime}} \otimes \ket{-_G}\!\bra{-_G}_{C} + \mathbb{I}_{C^{\prime}} \otimes \left( \mathbb{I}-\ket{-_G}\!\bra{-_G}\right) _{C} \right]\  \times \\
    \frac{1}{\sqrt{|G|}}\sum_{g \in G} \ket{0}_{C'}\ket{g}_{C} U(g) \ket{\psi} ,
\end{multline}
which in turn simplifies to 
\begin{multline}
    \ket{1}_{C^{\prime}} \ket{-_G}_C \Pi_{\mathcal{A}} \ket{\psi} - \ket{0}_{C^{\prime}} \ket{-_G}_C \Pi_{\mathcal{A}} \ket{\psi} + \\
    \ket{0}_{C^{\prime}} \frac{1}{\sqrt{|G|}} \sum_{g\in G} \ket{g}_C U(g) \ket{\psi} \, .
\end{multline}
This simplification can be computed directly from the definition of the antisymmetric projector given in~\eqref{eq:antisymmetricprojector} and an analogous version of~\eqref{eq:projectorassertion} with respect to the antisymmetric projector:
\begin{equation}
    (\bra{-_G}_C \otimes \mathbb{I})\frac{1}{\sqrt{|G|}} \sum_{g \in G} \ket{g}_C \otimes U(g) = \Pi_{\mathcal{A}} \, .
\end{equation}
Next, we uncompute the data register via the controlled unitary $\sum_{g \in G} \ket{g}\!\bra{g}_C \otimes U^{\dag}(g)$. The state is then given by
\begin{multline}
     \ket{1}_{C'} \frac{1}{\sqrt{|G|}} \sum_{g \in G} \operatorname{sgn}(g) \ket{g} U^\dagger(g) \Pi_{\mathcal{A}} \ket{\psi} \\
     - \ket{0}_{C'} \frac{1}{\sqrt{|G|}} \sum_{g \in G} \operatorname{sgn}(g) \ket{g} U^\dagger(g) \Pi_{\mathcal{A}} \ket{\psi} \\
     + \ket{0}_{C^{\prime}} \frac{1}{\sqrt{|G|}} \sum_{g\in G} \ket{g}_C \ket{\psi}.
\end{multline}
Now, using the fact that 
\begin{equation}
    \operatorname{sgn}(g)U(g) \Pi_{\mathcal{A}} = \operatorname{sgn}(g)U^\dag(g) \Pi_{\mathcal{A}} = \Pi_{\mathcal{A}},
\end{equation}
the state simplifies to
\begin{equation}
    \ket{1}_{C'} \ket{+_G}_C \Pi_{\mathcal{A}} \ket{\psi} + 
    \ket{0}_{C'} \ket{+_G}_C (\mathbb{I} - \Pi_{\mathcal{A}}) \ket{\psi}.
\end{equation}
After this, we discard the control register $C$. Then the state of the system is
\begin{equation}
    \ket{1}_{C^{\prime}} \Pi_{\mathcal{A}} \ket{\psi} + \ket{0}_{C^{\prime}} (\mathbb{I}-\Pi_{\mathcal{A}})\ket{\psi} \, .
\end{equation}

We can now see that, if we measure register $C'$ in the standard basis, then we will detect  $C'$ to be in the state~$\ket{1}$ with probability
\begin{equation}
    \Pr(\ket{1}_{C'}) = \left\Vert \Pi_{\mathcal{A}} \ket{\psi} \right\Vert_2^2 \, ,
\end{equation}
and with post-measurement state
\begin{equation}
    \frac{\Pi_{\mathcal{A}}\ket{\psi}}{\left\Vert \Pi_{\mathcal{A}} \ket{\psi} \right\Vert_2^2}.
\end{equation}
Furthermore, we detect the state $\ket{0}$ with probability
\begin{equation}\label{eq:asymmetric}
    \Pr(\ket{0}_{C'}) = \left\Vert \left( \mathbb{I} - \Pi_{\mathcal{A}} \right) \ket{\psi} \right\Vert_2^2 \, ,
\end{equation}
and with post-measurement state
\begin{equation}
    \frac{\left( \mathbb{I} - \Pi_{\mathcal{A}} \right) \ket{\psi}}{\left\Vert \left( \mathbb{I} - \Pi_{\mathcal{A}} \right) \ket{\psi} \right\Vert_2^2}.
\end{equation}

\subsection{Generalizing the Construction}\label{sec:genconcat}

Let us now attempt to further generalize our previous constructions. In the prior sections, we were able to realize projections of a state onto a symmetric subspace by uncomputing the circuit used to realize the projector; however, we can imagine using the same procedure to combine the action of two different group projectors. We will refer to this procedure throughout as concatenating the projectors, in reference to how the projectors are essentially performed in series within a single algorithm. The projectors should act on the same number of qubits, but a projector that acts on a smaller system can essentially be padded to a larger one via tensoring with the identity.

Suppose that we have a unitary representation $\{V(h)\}_{h \in H}$ of another group $H$. We can then follow the exact procedure from Section~\ref{sec:uncomputing} only replacing the superposition state with 
\begin{equation}
    \ket{+_H}_C \coloneqq \frac{1}{\sqrt{|H|}}\sum_{h \in H}\ket{h}_C \, .
\end{equation}
Then the first controlled unitary, used in step two, becomes
\begin{equation}
    \sum_{h \in H}\ket{h}\!\bra{h}_C \otimes V(h) \, ,
\end{equation}
and the final controlled unitary, used in step four, becomes
\begin{equation}
    \sum_{h \in H}\ket{h}\!\bra{h}_C \otimes V^\dag(h) \, .
\end{equation}
Implementing these substitutions, the resultant state before measurement is
\begin{equation}
    \ket{1}_{C'}\Pi_H\ket{\psi} + \ket{0}_{C'} \left(\mathbb{I} - \Pi_H \right) \ket{\psi} \, ,
\end{equation}
where the projection $\Pi_H$ is defined as
\begin{equation}
    \Pi_H \coloneqq \frac{1}{|H|}\sum_{h \in H} V(h) \, .
\end{equation}

Let us now consider cascading the two procedures. In this construction, depicted in Figure~\ref{fig:TwoProjectorTestGeneral}, we first perform the approach for the group $G$ and then follow it with the circuit corresponding to $H$. To achieve this, two ancilla qubits are necessary, which we denote $C'_1$ and $C'_2$, respectively. (A note of caution: the group representations should act on Hilbert spaces of the same size, and if $|G| \neq |H|$, the control state on $C$ will have to be discarded after the first algorithm and cannot be reused.) The transformation realized by this concatenation is
\begin{multline}
    \ket{\psi} \to \ket{1}_{C'_1}\ket{1}_{C'_2} \Pi_H \Pi_G \ket{\psi} + \ket{1}_{C'_1}\ket{0}_{C'_2} \Pi_H \left( \mathbb{I} - \Pi_G \right)\ket{\psi} \\
    + \ket{0}_{C'_1}\ket{1}_{C'_2}\left( \mathbb{I} - \Pi_H \right) \Pi_G \ket{\psi} \\
    + \ket{0}_{C'_1}\ket{0}_{C'_2}\left( \mathbb{I} - \Pi_H \right) \left( \mathbb{I} - \Pi_G \right) \ket{\psi} \, .
\end{multline}

\begin{figure}
\begin{center}
\includegraphics[width=\columnwidth]{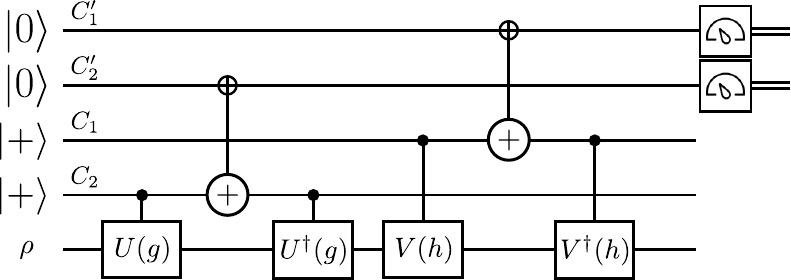}
\end{center}
\caption{Quantum circuit demonstrating the concatenation of two projectors in sequence. $\{U(g)\}$ is a unitary representation of one group and $\{V(h)\}$ is a unitary representation of another. Both unitary representations should act on spaces of the same dimension. 
}
\label{fig:TwoProjectorTestGeneral}
\end{figure}

\subsection{Examples}
\label{sec:examples}

\subsubsection{Concatenating Symmetric and Antisymmetric Projectors}

For a pertinent example, consider concatenating the approach from Section~\ref{sec:genconcat} for the symmetric and antisymmetric projectors of a group~$G$. Then $H = G$, and the two unitary representations can both be taken to be $\{U(\sigma)\}_{\sigma \in G}$. This allows us to use the same control register $C$ for both projectors. Then we can construct a simplified circuit as in Figure~\ref{fig:simplifiedSymmandAnti}, because the two controlled unitaries in the middle of Figure~\ref{fig:TwoProjectorTestGeneral} are inverses of each other and therefore cancel out (keeping in mind that there is a single control register $C$ in this special case). To illustrate this, we begin with an initial state given by
\begin{equation}\label{eq:initialsymantistate}
    \ket{00}_{C'C''}\ket{+_G}_C \ket{\psi} = \frac{1}{\sqrt{|G|}} \sum_{g\in G} \ket{00}_{C'C''}\ket{g}_C \ket{\psi} \, .
\end{equation}
In the second step we implement the controlled unitary $\sum_{g \in G} \ket{g}\!\bra{g}_C \otimes U(g)$ and the $\ket{+_G}_C$ controlled gate successively. This step realizes the state
\begin{align}
    \frac{1}{\sqrt{|G|}} &\sum_{g \in G} \ket{00}_{C'C''} \ket{g}_C U(g) \ket{\psi} \notag \\
    &\to \Bigg[ \big(X_{C'} \otimes \ket{+_G}\!\bra{+_G}_C \big) + \big( \mathbb{I}_{C'} \otimes (\mathbb{I}_{C}-\ket{+_G}\!\bra{+_G}_C) \big) \Bigg] \notag \\
    &\qquad \qquad  \frac{1}{\sqrt{|G|}} \sum_{g\in G} \ket{00}_{C'C''}\ket{g}_C U(g) \ket{\psi}  \\
    &= \ket{0}_{C''} \Bigg[ \ket{0}_{C'} \left ( \frac{1}{\sqrt{|G|}}\sum_{g\in G} \ket{g} U(g) \ket{\psi} - \ket{+_G} \Pi_\mathcal{S} \ket{\psi} \right) \notag \\
    &\qquad \qquad \qquad \qquad + \ket{1}_{C'} \ket{+_G}\Pi_\mathcal{S} \ket{\psi} \Bigg]. 
\end{align}
Note that the controlled-NOT gate acting on registers $C$ and~$C'$ essentially coherently couples the projection onto the symmetric subspace with the ancilla in register $C'$. The next step will act similarly with regard to the antisymmetric projection. By performing the $\ket{-_G}$ controlled gate from register $C$ to register $C''$, the state is now
\begin{align}
    &\ket{00}_{C'C''}  \Bigg( \innerprod{-_G}{+_G} \ket{-_G}_C \Pi_\mathcal{S} - \ket{+_G}_C\Pi_\mathcal{S} \notag\\
    &\qquad \qquad + \frac{1}{\sqrt{G}} \sum_{g\in G} \ket{g}_C U(g) - \ket{-_G}_C\Pi_\mathcal{A} \Bigg )\ket{\psi} \notag\\
    &+ \ket{01}_{C'C''} \ket{-_G}_C \left(\Pi_\mathcal{A}- \innerprod{-_G}{+_G} \Pi_\mathcal{S} \right) \ket{\psi} \notag\\
    &+\ket{10}_{C'C''} \left(\ket{+_G}_C \Pi_\mathcal{S} - \innerprod{-_G}{+_G} \ket{-_G}_C \Pi_\mathcal{S}\right) \ket{\psi} \notag \\
    &+\ket{11}_{C'C''} \ket{-_G}_C \innerprod{-_G}{+_G} \Pi_\mathcal{A}\Pi_\mathcal{S} \ket{\psi}.
\end{align}
Next define the following quantity 
\begin{equation}\label{eq:minusplusprod}
    G_p \coloneqq \innerprod{-_G}{+_G} = \frac{1}{|G|}\sum_{g \in G} \operatorname{sgn}(g) \, .
\end{equation}
Note that this quantity is equal to zero for the symmetric group $S_k$, but it is non-zero in some cases, e.g., for the cyclic group $Z_3$, for which every group element corresponds to an even permutation.

Finally, applying the last controlled unitary $\sum_{g \in G} \ket{g}\!\bra{g}_C \otimes U^\dagger(g)$, we get the following state:
\begin{multline}
\label{eq:final_state_+-}
    \ket{00}_{C'C''} \left[ \ket{+_G} \mathbb{I} - \ket{+_G} \Pi_\mathcal{A} -\ket{+_G} \Pi_\mathcal{S} + G_p \ket{-_G} \Pi_\mathcal{S}\right] \ket{\psi} \\
    + \ket{01}_{C'C''} \left[ \ket{+_G} \Pi_\mathcal{S} - G_p \ket{-_G} \Pi_\mathcal{S}\right] \ket{\psi} \\
     \qquad +\ket{10}_{C'C''} \left[ \ket{+_G} \Pi_\mathcal{A} - G_p \ket{-_G} \Pi_\mathcal{S}\right] \ket{\psi} \\
    + \ket{11}_{C'C''} \left[ G_p \ket{-_G} \Pi_\mathcal{S}\right] \ket{\psi}.
\end{multline}

In the case that $G$ is the symmetric group $S_k$, it follows that $G_p = \innerprod{-_G}{+_G} = 0$. Substituting for $G_p$, we find that the final state is then given by
\begin{multline}
    \ket{00}_{C'C''} \ket{+_G} \left[ \mathbb{I} - \Pi_\mathcal{A} - \Pi_\mathcal{S} \right] \ket{\psi} \\
    + \ket{01}_{C'C''} \ket{+_G} \Pi_\mathcal{S}  \ket{\psi} \\
     \qquad +\ket{10}_{C'C''}  \ket{+_G} \Pi_\mathcal{A}  \ket{\psi}.
     \label{eq:final-state-sym-anti}
\end{multline}
Lastly, also in the case that $G = S_k$, we find the following probabilities upon measuring the state of the $C'$ and $C''$ qubits:
\begin{align}
    P(\ket{00}_{C'C''}) &= \left \Vert (\mathbb{I}-\Pi_{\mathcal{S}}-\Pi_{\mathcal{A}})\ket{\psi} \right \Vert_2^2\, ,\label{eq:final-probs-sym-anti-1}\\
    P(\ket{01}_{C'C''}) &=  \left \Vert\Pi_{\mathcal{A}} \ket{\psi} \right \Vert_2^2 \, ,\\
    P(\ket{10}_{C'C''}) &=  \left \Vert\Pi_{\mathcal{S}} \ket{\psi} \right \Vert_2^2 \, ,\\
    P(\ket{11}_{C'C''}) &= 0  \, .
    \label{eq:final-probs-sym-anti-4}
\end{align}
Note that the probability of measuring the state $\ket{11}$ corresponds to measuring $\Pi_{\mathcal{S}}\Pi_{\mathcal{A}}\ket{\psi}$, which is always equal to zero in the case that $G = S_k$.

The specific example of $S_3$, the symmetric group on three letters, is described at length in Section~\ref{sec:ex1}, the next section.

\begin{figure}
\begin{center}
\includegraphics[width=\columnwidth]{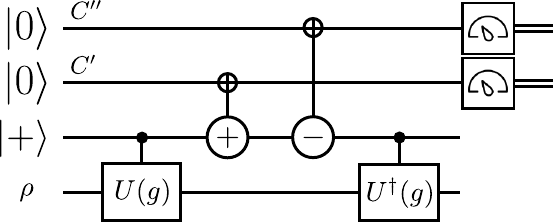}
\end{center}
\caption{The above circuit first performs the symmetric projection, followed by the antisymmetric projection. }
\label{fig:simplifiedSymmandAnti}
\end{figure}

\subsubsection{Three Irreps of \texorpdfstring{$S_3$}{S3}}

\label{sec:ex1} 

In this example, we recover the projections onto each of the three irreps of the group $S_3$ by concatenating the symmetric and antisymmetric projection circuits. As discussed in the previous section, both projectors utilize the same representations, $\{U(g)\}_{g \in G} = \{V(h)\}_{h \in H} $. This allows for a simplification of the middle portion of the algorithm, where 
\begin{equation}
    U^\dag(g) V(h) = U^\dag (g) U (g) = \mathbb{I} \,,
\end{equation}
as $h = g$ for this specific combination of projectors as long as the same initial control state $\ket{+_G}$ is used. This simplifies the resultant circuit, as the intermediary unitaries cancel. Indeed, such a simplification can be employed anytime the two projectors utilize the same unitaries in their representation. Since $S_3$ is a particular symmetric group, it follows that the final state before measurement is given by~\eqref{eq:final-state-sym-anti} and the final measurement probabilities are given by~\eqref{eq:final-probs-sym-anti-1}--\eqref{eq:final-probs-sym-anti-4}.

It is worth remembering, at this juncture, that all of these algorithms are necessarily dependent on the representation that is employed, even though this is often suppressed in the literature. For example, if we take the representation of $S_3$ that maps each transposition to a SWAP gate between two-dimensional Hilbert spaces (i.e., $(1\ 2 ) \to \operatorname{SWAP}_{A_1A_2}$), then the antisymmetric projector here is exactly zero. Indeed, there are no three-qubit states that demonstrate such antisymmetry. To guarantee a nontrivial antisymmetric subspace under these types of representations, we require that the dimension $d$ of the individual systems is greater than or equal to the number of systems to be permuted, $n$, because the size of the antisymmetric subspace is then given by $\binom{d}{n}$. Thus, using qudits, or systems of more than one qubit with controlled-SWAPs replaced by a concatenation of several controlled-SWAPs with the same control qubit, would give a nontrivial antisymmetric projector. 

However, other choices of representations can also yield nontrivial results. We choose a two-qubit unitary representation to demonstrate this in the example below.
In principle, any group representation that respects the group action and the sign representation should be adaptable to these algorithms.
Consider the representation of $S_3$ on two qubits that maps the transposition $(1\ 2 ) $ to the Hadamard gate on both qubits and $(2\ 3 )$ to applying a SWAP gate, CNOT gate, and then the Hadamard gate on both qubits. This is equivalent to mapping $ (1\ 2\ 3)$  to the operation $\operatorname{CNOT} \cdot \operatorname{SWAP}$, as shown in Table~\ref{tab:exampleS3}. A quick check can verify that this obeys the group action of~$S_3$. Then the symmetric projector is
\begin{equation}
    \Pi_{\mathcal{S}}=   \frac{1}{4}  
    \begin{pmatrix}
3 & 1 & 1 & 1 \\
1 & \sfrac{1}{3} & \sfrac{1}{3} & \sfrac{1}{3} \\
1 & \sfrac{1}{3} & \sfrac{1}{3} & \sfrac{1}{3} \\
1 & \sfrac{1}{3} & \sfrac{1}{3} & \sfrac{1}{3} \\
\end{pmatrix}\, , 
\end{equation}
and the antisymmetric projector is 
\begin{equation}
    \Pi_{\mathcal{A}}=   \frac{1}{4}  
    \begin{pmatrix}
1 & -1 & -1 & -1 \\
-1 & 1 & 1 & 1 \\
-1 & 1 & 1 & 1 \\
-1 & 1 & 1 & 1 \\
\end{pmatrix}\, . 
\end{equation}

\begin{table}
    \centering
    \begin{tabular}{P{1.5cm} | P{3.5cm}}
    \hline
         $\sigma \in S_3$ & $U(\sigma)$  \\ \hline \hline
         e & $\mathbb{I}$ \\ \hline
         $(1\ 2)$ & $H \otimes H$ \\ \hline 
         $(2\ 3)$ & $(H \otimes H) \operatorname{CNOT} \cdot \operatorname{SWAP}$ \\ \hline
         $(1\ 3)$ & $(H \otimes H) \operatorname{SWAP}\cdot \operatorname{CNOT}$ \\ \hline
         $(1\ 2\ 3)$ & $\operatorname{CNOT}\cdot \operatorname{SWAP}$ \\ \hline
         $(1\ 3\ 2)$ & $\operatorname{SWAP}\cdot \operatorname{CNOT}$ \\ \hline
    \end{tabular}
    \caption{An example representation of $S_3$ on two qubits.}
    \label{tab:exampleS3}
\end{table}

Thus, in this example, we can estimate the projection of the state $\ket{\psi}$ onto any of the three irreps of $S_3$. For larger symmetric groups, additional projectors could be added to sequentially project onto different subspaces.

\subsubsection{Difference of Two Projectors}
\label{ex:difference}

Consider a unitary that can be written as the difference of two projectors, $U = P - Q$, where $P$ and $Q$ are projectors. A specific example of this is the SWAP gate, which can be written as SWAP $= \Pi_\mathcal{S} - \Pi_\mathcal{A}$ and which we will discuss at more length in Section~\ref{sec:werner} in the context of determining Werner-state asymmetry of a bipartite state. Here, we will argue that such a unitary can be used to estimate the quantity $\left \Vert \operatorname{Re}(P \rho Q)\right \Vert_2^2$ for some state of interest $\rho$ and that the circuit used to do so can be constructed using our antisymmetry test perspective.

Suppose we have a completeness relation $\mathbb{I} = P + Q$; then, by simple substitution, we can see that $P = \frac{1}{2}(\mathbb{I} + U)$. We can then realize this projector by performing a controlled unitary. We use this fact to construct the circuit given in Figure~\ref{fig:Diff-Proj}, which we claim can estimate the quantity of interest, $\left \Vert \operatorname{Re}(P \rho Q)\right \Vert_2^2$.

\begin{figure}
\begin{center}
\includegraphics[width=\columnwidth]{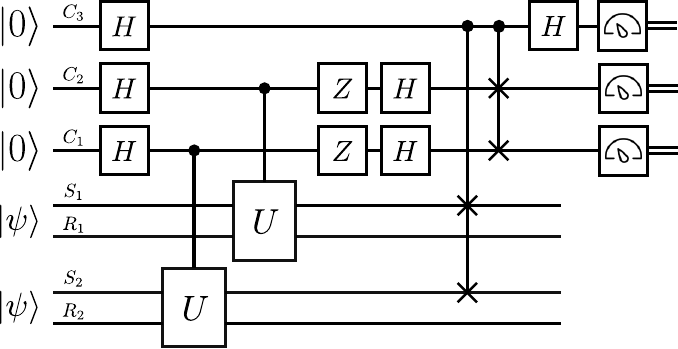}
\end{center}
\caption{Construction to estimate $\left \Vert \operatorname{Re}(P \rho Q)\right \Vert_2^2$, where $U=P-Q$.
}
\label{fig:Diff-Proj}
\end{figure}

This circuit uses an antisymmetry test as a subroutine. Consider the set $\{V(h)\}_{h\in H} = \{\mathbb{I}\otimes\mathbb{I},\mathbb{I}\otimes U, U \otimes \mathbb{I}, U \otimes U\}$ where $H=\{00,01,10,11\}$ corresponds to the necessary control-register states. Although we have not made the assumptions necessary to ensure that $\{V(h)\}_{h\in H}$ is a group representation, we can construct the circuit using the same principles regardless. For the construction, we assume that the state of interest $\rho_S$ is provided in the form of a purification  $\ket{\psi}\!\bra{\psi}_{SR}$, i.e., so that $\rho_S = \operatorname{Tr}_{R}[\ket{\psi}\!\bra{\psi}_{SR}]$.

The initial state of the system is given by
\begin{align}
\ket{+, +, +}_{C_1C_2C_3} \ket{\psi}_{S_1 R_1}\ket{\psi}_{S_2 R_2} \, ,
\end{align}
with $\ket{+}= \frac{1}{\sqrt{2}} (\ket{0}+\ket{1})$. However, we would like to frame this algorithm as an antisymmetry test, and so we instead write the initial state as
\begin{align}
    \ket{+}_{C_3}\ket{+_H}\ket{\Psi} \, ,
\end{align}
where $\ket{\Psi} \coloneqq  \ket{\psi}_{S_1 R_1}\otimes \ket{\psi}_{S_2 R_2}$ and 
\begin{equation}
    \ket{+_H} = \frac{1}{\sqrt{|H|}}\sum_{h\in H}\ket{h}\, .
\end{equation}
The action of the antisymmetric subroutine is thus as follows:
\begin{align}
    \ket{+}_{C_3}\ket{+_H}\ket{\Psi} &
    &\to \frac{\ket{+}_{C_3}}{2}\sum_{h\in H}\operatorname{sgn}(h)\ket{h} V(h) \ket{\Psi}\, , \label{eq:pqstep1}
\end{align}
noting that $|H| = 4$ by construction. Substituting the definitions above into~\eqref{eq:pqstep1} gives us
\begin{multline}
        = \frac{\ket{+}_{C_3}}{2} \Big( \ket{00} ( \ket{\psi}_{S_1R_1} \otimes\ket{\psi}_{S_2R_2}) 
        \\ - \ket{01} ( \ket{\psi}_{S_1R_1} \otimes U\ket{\psi}_{S_2R_2} ) 
        \\ - \ket{10} (U \ket{\psi}_{S_1R_1} \otimes \ket{\psi}_{S_2R_2}) 
        \\ + \ket{11} ( U \ket{\psi}_{S_1R_1} \otimes U\ket{\psi}_{S_2R_2} )  \Big) \, .
\end{multline}
Next, we act on qubits $C_1$ and $C_2$ with Hadamard gates, seen in Figure~\ref{fig:Diff-Proj}, leading to the state:
\begin{multline}\label{eq:pqstep2}
    \frac{\ket{+}_{C_3}}{4} \Bigg( \Big( \ket{0}_{C_1} ( \mathbb{I} - U) \ket{\psi}_{S_1R_1} + \ket{1}_{C_1}( \mathbb{I} + U) \ket{\psi}_{S_1R_1}) \Big)
        \\\otimes \Big( \ket{0}_{C_2} ( \mathbb{I} - U) \ket{\psi}_{S_2R_2} + \ket{1}_{C_2}( \mathbb{I} + U) \ket{\psi}_{S_2R_2} \Big) \Bigg) \, .
\end{multline}
Recall that $P = \frac{1}{2}\left(\mathbb{I} + U\right)$ and thus by the completeness relation, $Q= \frac{1}{2}\left(\mathbb{I} - U\right)$. This allows us to rewrite~\eqref{eq:pqstep2} as
\begin{multline}\label{eq:ex2aftersubroutine}
    \ket{+}_{C_3}  \otimes  \left ( \ket{0}_{C_1} Q \ket{\psi} + \ket{1}_{C_1} P \ket{\psi}\right ) \\
    \otimes \left ( \ket{0}_{C_2} Q \ket{\psi} + \ket{1}_{C_2} P \ket{\psi} \right ) .
\end{multline}

Now we move on to the second half of the circuit, which realizes a SWAP test between the systems $S_1$ and $S_2$ using $C_3$ as the control. The effect of these controlled-SWAPs on the state of~\eqref{eq:ex2aftersubroutine} is given by
\begin{align}
    &\frac{1}{\sqrt{2}} \Big[ \ket{000} [Q\ket{\psi} \otimes Q\ket{\psi}] + \ket{001} [Q\ket{\psi} \otimes P\ket{\psi}] \notag \\
    &\quad + \ket{010} [P\ket{\psi} \otimes Q\ket{\psi}] + \ket{011} [P\ket{\psi} \otimes P\ket{\psi}] \notag \\
    &\quad + \ket{100} F[Q\ket{\psi} \otimes Q\ket{\psi}] + \ket{101} F[Q\ket{\psi} \otimes P\ket{\psi}] \notag \\
    &\quad + \ket{110} F[P\ket{\psi} \otimes Q\ket{\psi}] + \ket{111} F[P\ket{\psi} \otimes P\ket{\psi}]\, ,
\end{align}
with $F \coloneqq F_{S_1 S_2}$ the unitary swap operator.
Next, we apply the final controlled-SWAP gate from $C_3$ to $C_1$ and $C_2$ and the Hadamard gate on $C_3$. The state just before the final measurements is then
\begin{align}
\label{eq:diff-prof-pre-trace}
    \ket{\Phi} &\coloneqq  \frac{1}{2} \Big[ 
              \ket{000} [Q\ket{\psi} \otimes Q\ket{\psi} + F[Q\ket{\psi} \otimes Q\ket{\psi}]] \notag \\
    &\qquad + \ket{001} [Q\ket{\psi} \otimes P\ket{\psi} + F[P\ket{\psi} \otimes Q\ket{\psi}]] \notag \\
    &\qquad + \ket{010} [P\ket{\psi} \otimes Q\ket{\psi} + F[Q\ket{\psi} \otimes P\ket{\psi}]] \notag \\
    &\qquad + \ket{011} [P\ket{\psi} \otimes P\ket{\psi} + F[P\ket{\psi} \otimes P\ket{\psi}]] \notag \\
    &\qquad + \ket{100} [Q\ket{\psi} \otimes Q\ket{\psi} - F[Q\ket{\psi} \otimes Q\ket{\psi}]] \notag \\
    &\qquad + \ket{101} [Q\ket{\psi} \otimes P\ket{\psi} - F[P\ket{\psi} \otimes Q\ket{\psi}]] \notag \\
    &\qquad + \ket{110} [P\ket{\psi} \otimes Q\ket{\psi} - F[Q\ket{\psi} \otimes P\ket{\psi}]] \notag \\
    &\qquad + \ket{111} [P\ket{\psi} \otimes P\ket{\psi} - F[P\ket{\psi} \otimes P\ket{\psi}]] \Big].
\end{align}
This state can be written more compactly as 
\begin{multline}
    \ket{\Phi} = \frac{1}{2} \sum_{a,b,c\in\{0,1\}} \ket{abc} \left[ P_b \ket{\psi} \otimes P_c \ket{\psi} \right.\\
    \left. + (-1)^a F[P_c \ket{\psi} \otimes P_b \ket{\psi}]\right],
\end{multline}
where we define $P_0 \coloneqq Q$ and $P_1 \coloneqq P$. Measuring the first three qubits $C_3 C_2 C_1$, the probability $p(abc)$ of measuring the bitstring $abc$ is given by
\begin{align}
    & p(abc) \notag \\
    &= \left\Vert (\bra{abc}_{C_3 C_2 C_1} \otimes \mathbb{I}_{S_1 R_1 S_2 R_2}) \ket{\Phi} \right\Vert^2_2 \\
    &= \frac{1}{2} [\Tr{P_b \rho}\Tr{P_c \rho} + (-1)^a \Tr{P_b \rho P_c \rho}].
\end{align}
Thus, we can take the difference between some of the probabilities  to find that 
\begin{align}
    p(000) - p(100) &= \Tr{Q\rho Q\rho}, \label{eq:p000-p100} \\
    p(001) - p(101) &= \Tr{Q\rho P\rho}, \label{eq:p001-p101} \\
    p(010) - p(110) &= \Tr{P\rho Q\rho} = \Tr{Q\rho P\rho}, \label{eq:p010-p110}  \\
    p(011) - p(111) &= \Tr{P\rho P\rho}.
    \label{eq:p011-p111}
\end{align}
Next, we show that this is simply related to the quantity of interest: $\left \Vert \operatorname{Re}(P \rho Q)\right \Vert_2^2$. Consider that
\begin{equation}
    \left \Vert \operatorname{Re}(X) \right\Vert^2_2 = \frac{1}{4} \operatorname{Tr}[X^2 + XX^\dagger + X^\dagger X + (X^\dagger)^2].
\end{equation}
Substituting for $X$ with $P \rho Q$, we find that
\begin{align}
    &\left\Vert \operatorname{Re}(P \rho Q) \right\Vert_2^2 \notag \\
    &= \frac{1}{4} \Big[ \Tr{QP\rho QP\rho} + \Tr{PQ\rho PQ\rho} + 2\Tr{P\rho Q\rho} \Big] \notag \\
    &= \frac{1}{2} \Tr{P\rho Q\rho},
    \label{eq:real-part-norm-calc}
\end{align}
where the second equality results from the fact that $PQ=QP=0$. 
Combining both results, we see that
\begin{align}
    \left \Vert \operatorname{Re}(P \rho Q)\right \Vert_2^2 & = \frac{1}{2} (p(001) - p(101))\\
    & = \frac{1}{2} (p(010) - p(110)).
\end{align}

Thus, by repeating the circuit multiple times and building measurement statistics, we can estimate $\left \Vert \operatorname{Re}(P \rho Q)\right \Vert_2^2$. Additionally, by collecting both the measurement results from~\eqref{eq:p001-p101} and~\eqref{eq:p010-p110}, we can converge to the desired quantity of $\left \Vert \operatorname{Re}(P \rho Q)\right \Vert_2^2$ faster than if we were sampling only a single measurement outcome. In more detail, for all $t\in \{1, \ldots, T\}$, let $Y_t$ be a random variable that is set to $0$ after obtaining the results 000, 100, 011, or 111, $+1$ after obtaining the results 001 or 010, and $-1$ after obtaining the results 101 or 110. Then, it follows from~\eqref{eq:p000-p100}--\eqref{eq:p011-p111} and~\eqref{eq:real-part-norm-calc} that
\begin{equation}
\mu \equiv \mathbb{E}[Y_t] = 2 \Tr{P\rho Q\rho} = 4 \left \Vert \operatorname{Re}(P \rho Q)\right \Vert_2^2,
\label{eq:mu-mean}
\end{equation}
so that $\overline{Y_T} \coloneqq \frac{1}{T}\sum_{t=1}^T Y_t$ is an unbiased estimator of $4 \left \Vert \operatorname{Re}(P \rho Q)\right \Vert_2^2$. Set $\varepsilon > 0$ and $\delta \in (0,1)$. The Hoeffding bound (see, e.g., \cite[Lemma~1]{rethinasamy2023quantum})  implies that, if $T \geq \frac{2}{\varepsilon^2} \ln \!\left( \frac{2}{\delta}\right)$, then $\Pr[\vert \overline{Y_T} - \mu \vert \leq \varepsilon] \geq 1-\delta$.

\subsubsection{Resolution of Identity}

\label{ex:res_identity}

Consider the $r$-term resolution of identity, given by 
\begin{equation}
    \mathbb{I} = P_0 + \cdots + P_{r-1},
\end{equation}
where each $P_i$ is a projector.
In this example, we provide an algorithm, and an antisymmetry perspective, to estimate terms of the form 
\begin{equation}
    \left\Vert \operatorname{Re}(P_i \rho P_j) \right\Vert^2_2.
\end{equation}
For $i \neq j$, this quantity measures the component of the state~$\rho$ in the off-block-diagonal subspaces. For example, using projectors onto subspaces of fixed Hamming weight, it was shown that quantitites of this form can be used to test if a state is invariant under collective phase rotations \cite{laborde2021testing}. More generally, estimating terms of this form can measure coherences, or lack thereof, which is relevant in a range of settings, particularly in thermodynamics \cite{yuan2020direct,rodríguezrosario2013thermodynamicsquantumcoherence} and coherence distillation \cite{Liu2019Distillation,Regula2018distill}. 

Using a similar argument as in the previous section, it follows that, when $i \neq j$,
\begin{equation}
\label{eq:cross-term}
    \left\Vert \operatorname{Re}(P_i \rho P_j) \right\Vert^2_2 = \frac{1}{2} \Tr{P_i \rho P_j \rho}.
\end{equation}
Define the following unitary $U$:
\begin{equation}
\label{eq:def-U-res-I}
    U \coloneqq \sum_{j=0}^{r-1} \exp{\left(\frac{2\pi \mathrm{i} j}{r}\right)} P_j.
\end{equation}
We also define the following operation $\tilde{Z}$ as follows:
\begin{equation}
\label{eq:Ztilde}
    \tilde{Z} \coloneqq \sum_{c=0}^{r-1} \exp{\left(\frac{2\pi \mathrm{i} c}{r}\right)} \outerproj{c}.
\end{equation}

\begin{figure}
\begin{center}
\includegraphics[width=\columnwidth]{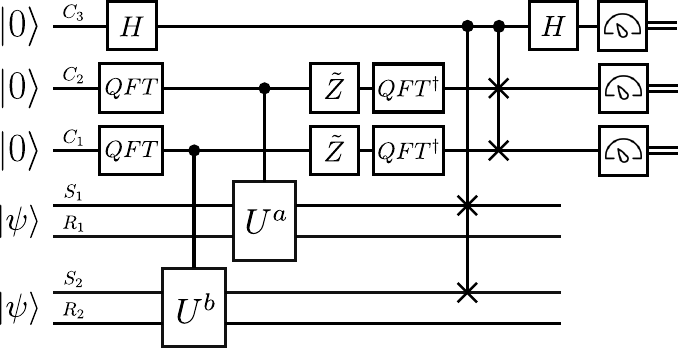}
\end{center}
\caption{Construction to estimate $\left \Vert \operatorname{Re}(P_i \rho P_j)\right \Vert_2^2$, where $U$ is defined in~\eqref{eq:def-U-res-I} and the unitary $\tilde{Z}$ is defined in~\eqref{eq:Ztilde}.}
\label{fig:Res-identity}
\end{figure}

The circuit uses an antisymmetry test as a subroutine. Consider the set
\begin{equation}
\{V(h)\}_{h \in H} = \{U^a \otimes U^b\ \vert\ a,b \in \{0, \ldots, r-1\}\},
\end{equation}
 where $H = \{a,b\ \vert\ a,b\in \{0, \ldots, r-1\}\}$,    
corresponding to the necessary control register states. Although we have not made the assumptions necessary to ensure $\{V(h)\}_{h\in H}$ is a group representation, we can construct the circuit using the same principles regardless. For the construction, we assume that the state of interest $\rho_S$ is provided in the form of a purification $\ket{\psi}\!\bra{\psi}_{SR}$, i.e., such that  $\rho_S = \operatorname{Tr}_{R}[\ket{\psi}\!\bra{\psi}_{SR}]$.

For the circuit, we append three ancilla registers --- $C_3$, which is a qubit, and $C_2$ and $C_1$, which are $r$-level systems. As seen in Figure~\ref{fig:Res-identity}, we measure Pauli-$Z$ on the first qubit and assign the value of $+1$ if the outcome zero occurs, and $-1$ if the outcome one occurs, and we measure the other two $r$-level systems in the computational basis. The different expected values are given by
\begin{equation}\label{eq:resolution-of-identity-end-result}
    p(0ab) - p(1ab) = \Tr{P_{a-1}\rho P_{b-1}\rho},
\end{equation}
where the subtraction in $a-1$ and $b-1$ is modulo $r$.
We provide the state simulation and the proof of this result in Appendix~\ref{app:proof-res-identity}. Thus, when $a\neq b$, this circuit measures the cross terms as desired. Furthermore, we see that this is a generalization of Example~\ref{ex:difference}, where we set $r=2$, $P_0 = P$ and $P_1 = Q$. Using~\eqref{eq:cross-term}, we can estimate all the cross terms from the statistics of the different measurement outcomes.

\section{Applications}
\label{sec:applications}

\subsection{Estimating Schmidt Rank}

The Schmidt rank of a pure state can be estimated by utilizing a projection onto the antisymmetric subspace~\cite{JLV22}. In this section, we will briefly review Schmidt rank and the aforementioned result before giving example computations of this application.

The Schmidt decomposition of a bipartite pure state is a superposition of coordinated orthonormal states. We state the theorem more formally as follows.
\begin{theorem}[Schmidt Decomposition]
Suppose that we have a bipartite pure state, \begin{equation}
    \vert \psi \rangle_{RS} \in \mathcal{H}_R \otimes \mathcal{H}_S,
\end{equation}
where $\mathcal{H}_R$ and $\mathcal{H}_S$ are finite-dimensional Hilbert spaces, not necessarily of the same dimension, and $\left\Vert \vert \psi \rangle_{RS} \right\Vert_2 = 1$. Then it is possible to express this state as follows: 
\begin{equation}
    \vert \psi \rangle_{RS} = \sum_{i=0}^{r-1} \lambda_i \vert v_i \rangle_R \vert w_i \rangle_S,
\end{equation}
where the amplitudes $\lambda_i$ are real, strictly positive, and normalized so that $\sum_i \lambda_i^2 =1$, the states $\{\vert v_i \rangle_R\}_{i=0}^{r-1}$ form an orthonormal basis in system $R$, and the states $\{\vert w_i \rangle_S\}_{i=0}^{r-1}$ form an orthonormal basis in system $S$. The Schmidt rank $r$ of a bipartite state is equal to the number of Schmidt coefficients in its Schmidt decomposition; i.e., $r = |\{\lambda_i\}_{i=0}^{r-1}|$.
\end{theorem}

The Schmidt decomposition theorem applies to any bipartite cut of a pure quantum state, and the Schmidt rank $r$ obeys the following inequality
\begin{equation}
    r \leq \min\{\operatorname{dim}(\mathcal{H}_R), \operatorname{dim}(\mathcal{H}_S)\}.
\end{equation}
It is well known that the Schmidt rank can be used to decide if a state is entangled. The Schmidt rank $r>1$ if and only if the state is entangled. Another application of estimating the Schmidt rank, in the context of quantum machine learning, can found in~\cite{JLRW23}.

The result from~\cite{JLV22} is as follows:
\begin{theorem}[\cite{JLV22}]
    Suppose $ \vert \psi \rangle_{RS} \in \mathcal{H}_R \otimes \mathcal{H}_S$ is a bipartite state vector. Then $\operatorname{SR}(\vert \psi \rangle) \leq r$ if and only if 
    \begin{equation}
        \left(\Pi^{\mathcal{A}}_{R, r+1} \otimes \Pi^{\mathcal{A}}_{S, r+1} \right) (\vert \psi \rangle_{RS}^{\otimes r+1}) = 0,
    \end{equation}
    where $\Pi^{\mathcal{A}}_{R, r+1}$ is the projection onto the antisymmetric subspace of $\mathcal{H}^{\otimes (r+1)}_R$ (and similarly for the $S$ system).
\end{theorem}

Using the antisymmetry test from Section~\ref{sec:antisymmetry}, we can construct a test for estimating the Schmidt rank by performing two antisymmetry tests in tandem. This process is shown in the schematic in Figure~\ref{fig:SchmidtSchematic}, and we give a concrete construction for $r=2$ in Figure~\ref{fig:SchmidtSchematic2}. For these algorithms, a non-zero probability of measuring the all-zeros ($\ket{\vec{0}}_C$) outcome indicates that the Schmidt rank of the state is strictly greater than $r$, i.e., $\operatorname{SR}(\vert \psi \rangle) > r$. 
Note that in both figures, the registers $R$ and $S$ can be multi-qubit registers. In that case, each qubit of the multi-qubit registers is individually swapped. The count of the number of qubits, two-qubit gates, and CSWAPs to compare the Schmidt rank of a bipartite $n$-qubit state, and a constant $r$ is given in Table~\ref{tab:schmidt-rank-count}. 

\begin{figure}[t]
\begin{center}
\includegraphics[width=\columnwidth]{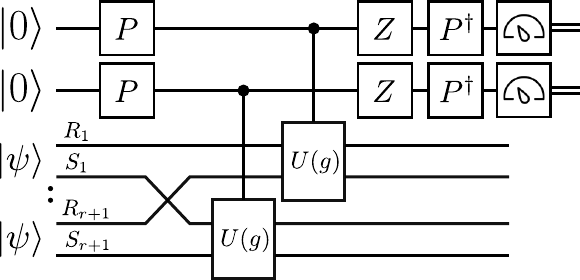}
\end{center}
\caption{Construction to test if the Schmidt rank of the input state is less than or equal to $r$. The circuit realizes two antisymmetric projections on the $R$ and $S$ subsystems independently. The unitary $P$ creates the appropriate superposition for the given $r$, and $U(g)$ is one of the corresponding group elements.}
\label{fig:SchmidtSchematic}
\end{figure}

\begin{figure}
\begin{center}
\includegraphics[width=\columnwidth]{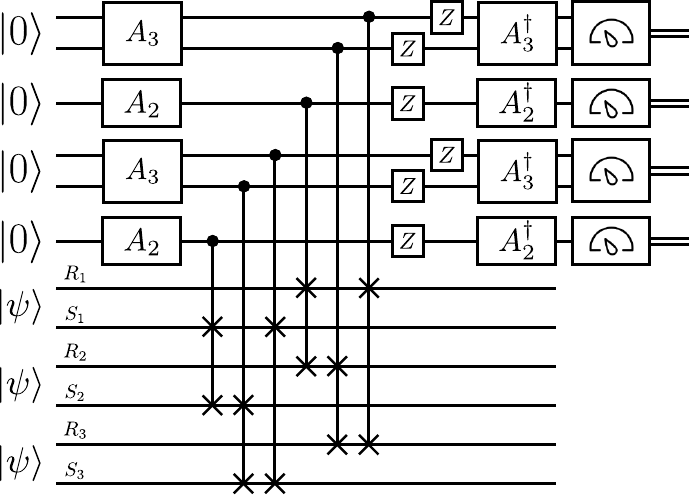}
\end{center}
\caption{Concrete example of Schmidt rank test for $r=2$. 
}
\label{fig:SchmidtSchematic2}
\end{figure}

\renewcommand{\arraystretch}{1.5}
\begin{table}
    \centering
    \begin{tabular}{P{3cm}|P{4cm}} \hline
    Property & Count \\ \hline\hline
    Qubits & $(r+1)(n + r)$\\
    Two-qubit gates & $4r(r-1)$ \\
    CSWAPs & $nr(r+1)$ \\
    \hline
    \end{tabular}
    \caption{Number of elements for the Schmidt rank estimation algorithm. Here $n$ denotes the total number of qubits of the state being tested, and $r$ denotes the Schmidt rank test value.}
    \label{tab:schmidt-rank-count}
\end{table}
\renewcommand{\arraystretch}{1.0}

We demonstrate this algorithm in Figures~\ref{fig:Schmidt_GHZ}, \ref{fig:Schmidt_W}, and \ref{fig:Schmidt_Bell_Bell} for a Bell state, W state, and double Bell state respectively. For each of these states, we plot the probability of obtaining the all-zeros outcome as a function of the Schmidt cutoff rank~$r$. These states are all known to be entangled, and that is reflected in the results given. Each of the plots have the probability $p(0) \neq 0$ for $r=1$, implying that their Schmidt rank $r >1$, the signature of an entangled state.  

\begin{figure}
    \centering
    \includegraphics[width=\columnwidth]{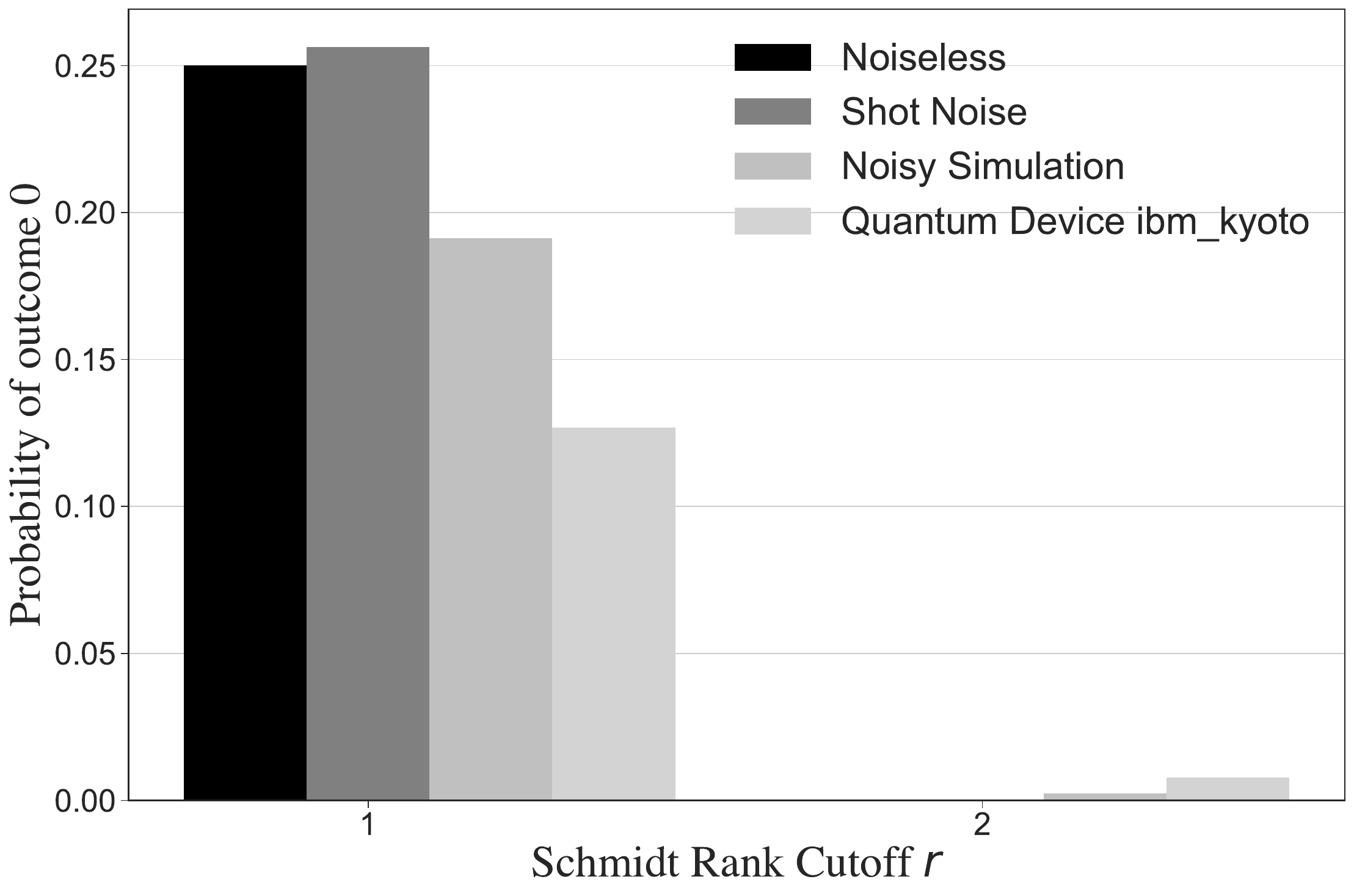}
    \caption{
    Measurement statistics for the Bell state $\ket{\Phi^+}_{RS}$. We see that the probability of obtaining outcome $0$ is non-zero for $r=1$ and is zero for $r=2$. This implies that the Schmidt rank of the state across the $R:S$ bipartition is two.}
    \label{fig:Schmidt_GHZ}
\end{figure}

\begin{figure}
    \centering
    \includegraphics[width=\columnwidth]{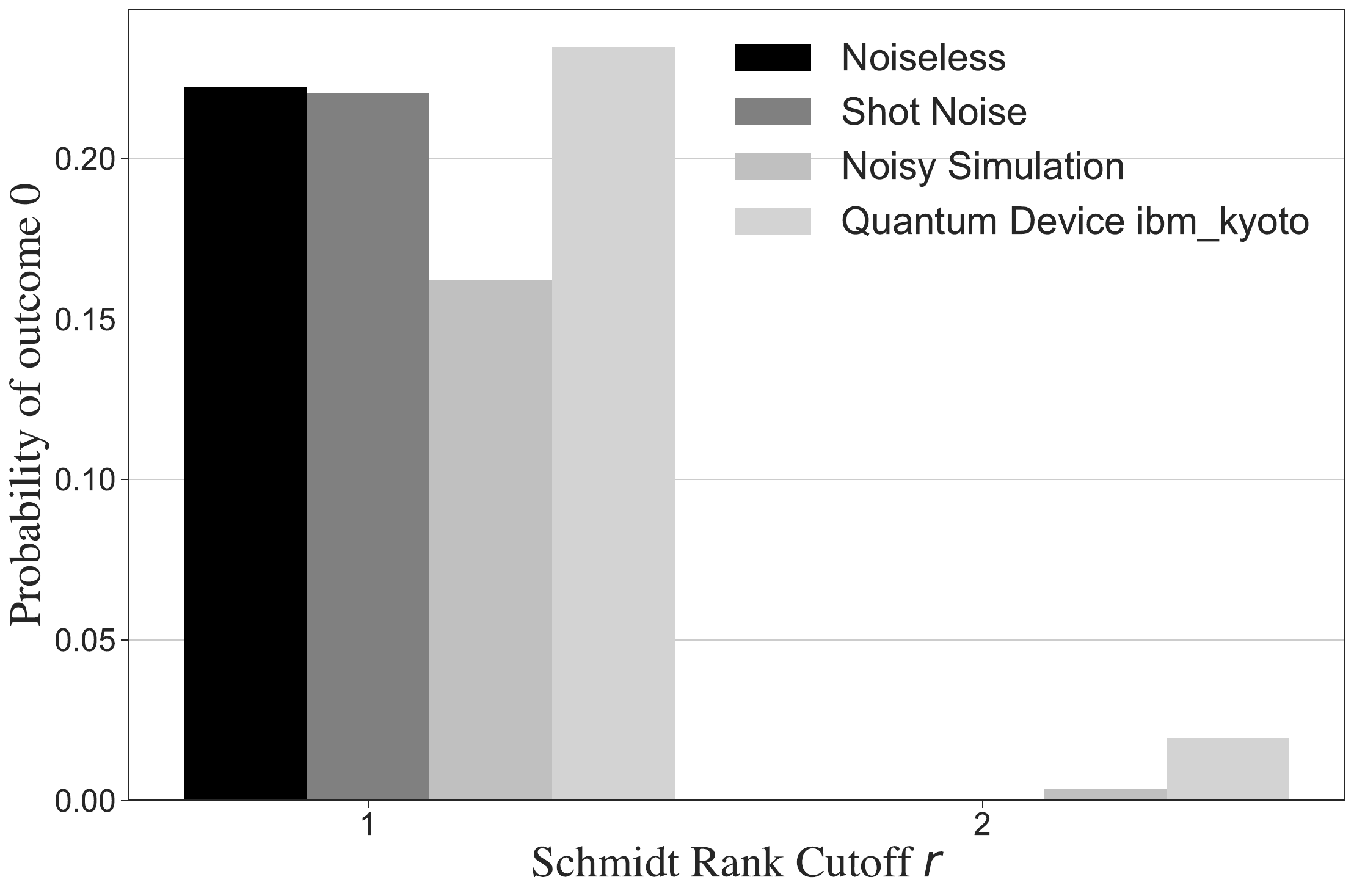}
    \caption{Measurement statistics for the W state $\ket{W}_{RS}$, where $\vert R \vert = 2$. 
    We see that the probability of obtaining outcome $0$ is non-zero for $r=1$, and is zero for $r=2$. This implies that the Schmidt rank of the state across the $R:S$ bipartition is two.}
    \label{fig:Schmidt_W}
\end{figure}

\begin{figure}
    \centering
    \includegraphics[width=\columnwidth]{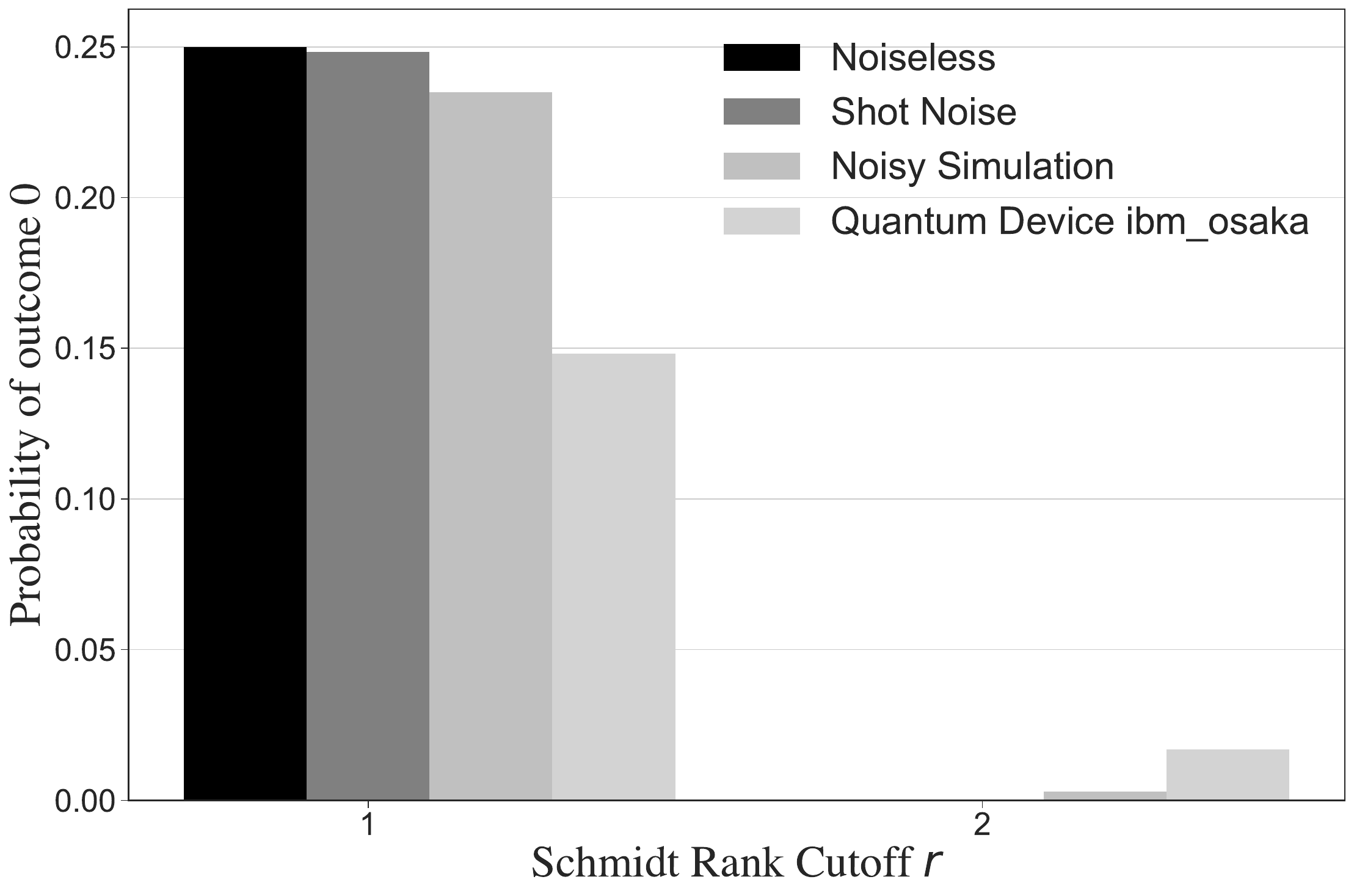}
    \caption{Measurement statistics for the double Bell state $\ket{\Phi^+}_{RS_1} \otimes \ket{ \Phi^+}_{S_2S_3}$, where $\vert R \vert = 1$. The bipartition here is between the first qubit and the last three.
    Since the second Bell state lies completely in the $S$ partition, it should have no effect on the algorithm, and the data indicates this fact. We see that the probability of obtaining outcome $0$ is non-zero for $r=1$ and is zero for $r=2$. This implies that the Schmidt rank of the state across the $R:S$ bipartition is two.}
    \label{fig:Schmidt_Bell_Bell}
\end{figure}

\subsection{Alternate Test for Werner States}\label{sec:werner}

A Werner state~\cite{Werner89} exhibits a particular kind of symmetry called Werner invariance. Consider a state $\rho$ on $N$ qudits, and then define $\overline{\rho}$ by
\begin{equation}
    \overline{\rho} = \int (U)^{\otimes N} \rho (U^\dag)^{\otimes N} dU \, .
\end{equation}
If a state $\rho$ exhibits symmetry under a collective unitary action, then $\rho=\overline{\rho}$ is called a Werner state. For instance, this is the case for a two-party state that satisfies the following:
\begin{equation}\label{eq:2qubitwerner}
    \overline{\rho} = \int (U \otimes U) \rho_{AB} (U \otimes U)^{\dag} dU \, .
\end{equation}

As a special case of Schur--Weyl duality, we know that $\rho_{AB} =\overline{\rho}$ in~\eqref{eq:2qubitwerner} if and only if
\begin{equation}\label{eq:werneralt}
    \rho_{AB} = p\frac{\Pi_S}{d(d+1)/2} + (1-p)\frac{\Pi_{\mathcal{A}}}{d(d-1)/2} \, ,
\end{equation}
where $p \in [0,1]$, $\Pi_S$ is the projection onto the symmetric subspace, and $\Pi_{\mathcal{A}}$ is the projection onto the antisymmetric subspace. Specifically, we define $\Pi_S$ and $\Pi_{\mathcal{A}}$ to be 
\begin{align}
    \Pi_S & \coloneqq \frac{\mathbb{I}_{AB} + F_{AB}}{2} \, , \\
    \Pi_{\mathcal{A}} & \coloneqq \frac{\mathbb{I}_{AB} - F_{AB}}{2}\, ,
\end{align}
where $F_{AB}$ is the unitary swap operator for systems $A$ and $B$.

Since $\Pi_S + \Pi_{\mathcal{A}} = \mathbb{I}_{AB}$, it follows that a general state can be decomposed as
\begin{align}
    \sigma_{AB} &= (\Pi_S + \Pi_{\mathcal{A}}) \sigma_{AB} (\Pi_S + \Pi_{\mathcal{A}}) \\
    &= \Pi_S \sigma_{AB} \Pi_S + \Pi_{\mathcal{A}}\sigma_{AB}\Pi_{\mathcal{A}} \notag \\
    & \qquad + \Pi_{\mathcal{A}}\sigma_{AB}\Pi_S + \Pi_S\sigma_{AB}\Pi_{\mathcal{A}}\, .
\end{align}
By comparison with~\eqref{eq:werneralt}, we see that a state $\sigma_{AB}$ is a Werner state if and only if both of the cross terms $\Pi_{\mathcal{A}}\sigma_{AB}\Pi_S$ and $\Pi_S\sigma_{AB}\Pi_{\mathcal{A}}$ are zero. Thus, it suffices to develop a quantum algorithm for estimating the norm of these matrix elements in order to perform Werner-state symmetry tests.

We can perform such a test by using the general approach developed in Example~\ref{ex:difference}. First, notice that $\operatorname{SWAP} =(\Pi_\mathcal{S}- \Pi_{\mathcal{A}})$. Then let $U=\operatorname{SWAP}$. By measuring the control state as described therein, we obtain an estimate of $\left \Vert  \operatorname{Re} (\Pi_{\mathcal{S}} \rho \Pi_{\mathcal{A}}) \right \Vert_2^2$. This is depicted in Figure~\ref{fig:altaltWernertest}.

A simple example of a Werner state  is the singlet state 
\begin{equation}\label{eq:singlet}
    \ket{\psi^-} = \frac{1}{\sqrt{2}}\left ( \ket{01}-\ket{10} \right ) \, .
\end{equation}
Indeed, every pure Werner state is equal to a linear combination of singlets~\cite[Theorem~1]{Lyons2008multiparty}.
For example, consider the ``pizza'' state for six qubits given in~\cite[Figure~1]{Lyons2023Werner}:
\begin{multline}\label{eq:pizza}
    \ket{P} \coloneqq \frac{1}{\sqrt{8}} \Big ( \ket{000111} - \ket{111000} - \ket{001110} \\ + \ket{110001} - \ket{010101} + \ket{101010} \\
    + \ket{011100} - \ket{100011} \Big ) \, .
\end{multline}
Another example of a two-qubit Werner state is achieved via a convex combination of the maximally-mixed state and a Bell state, 
\begin{equation}\label{eq:mixedWerner}
    W(p) = \frac{p}{6}
    \begin{pmatrix}
\frac{p}{3} & 0 & 0 & 0  \\
0 & \frac{3-2p}{6} & \frac{-3+4p}{6} & 0\\
0 & \frac{-3+4p}{6} & \frac{3-2p}{6} & 0\\
0 & 0 & 0 & \frac{p}{3}\\
\end{pmatrix}\, . 
\end{equation}

We include computational results from applying our algorithm to the states in~\eqref{eq:singlet},~\eqref{eq:pizza}, and~\eqref{eq:mixedWerner} in Table~\ref{tab:WernerStateData}. For comparison, we also include the results of using a random two-qubit state and the states $\ket{00}$ and $\ket{11}$, none of which demonstrate Werner invariance.

\begin{table*}[t]
\centering
\begin{tabular}{c|c|c|c|c|c}
    State $\rho_{AB}$ & True Value & Noiseless & Shot Noise & Noisy & Quantum Device \\
    \hline
    \hline
    $\ket{\psi^-}$ & $0.0$ & $0.0$ & $0.0$ & $-0.0003$ & $-0.007$\\
    $\ket{P}$ & $0.0$ & $0.0$ & $0.0$ & $-0.001$ & $0.0008$\\
    $W(0.3)$ & $0.0$ & $0.0$ & $-0.002$ & $-5.6 \times 10^{-5}$ & $-2.1 \times 10^{-5}$\\
    Random $2$-qubit & $0.195$ & $0.195$ & $0.196$ & $0.071$ & $0.004$ \\
    $\ket{00}$ & $0.0$ & $0.0$ & $0.0$ & $0.001$ & $-0.0002$ \\
    $\ket{11}$ & $0.0$ & $0.0$ & $0.0$ & $0.0002$ & $0.01$ \\
    \end{tabular}
    \caption{Results of test for Werner States, which estimates $\left\Vert \operatorname{Re}(\Pi_{\mathcal{S}} \rho \Pi_{\mathcal{A}}) \right\Vert^2_2$ for different states of interest. We provide the true value (calculated by classical means), noiseless, shot noise with $10^5$ shots, noisy simulations, and quantum device {\tt ibm\_osaka}.}
    \label{tab:WernerStateData}
\end{table*}

\begin{figure}
\begin{center}
\includegraphics[width=\columnwidth]{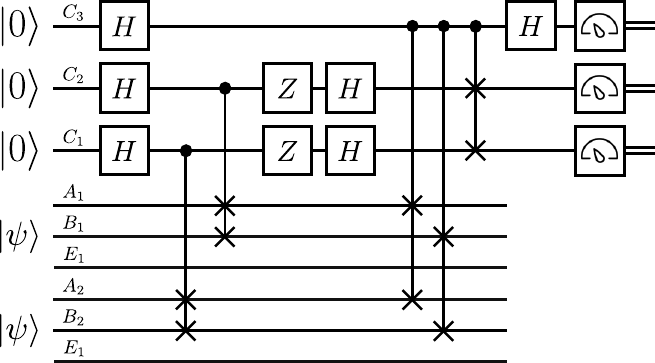}
\end{center}
\caption{Quantum circuit to test whether is a quantum state is a Werner state. 
}
\label{fig:altaltWernertest}
\end{figure}

\subsection{Estimating Commutators}

\label{sec:commutators}

Commutators play an important role in quantum mechanics, especially in quantifying uncertainty when measuring two observables. We present here two different methods for estimating a commutator and a nested commutator. We show how the antisymmetric circuit structure can be easily adapted to realize a quantum switch between operators, as described in~\cite{switch1,switch2,switch3}.

We begin by considering the commutator between two unitaries $A$ and $B$, which for now we suppose have the same dimensions. Figure~\ref{fig:commutator} depicts a quantum circuit to obtain this quantity. First, consider an input state $\ket{0}_C\ket{\psi}$. After acting on the control qubit with a Hadamard and then a Pauli-$Z$ gate, the state of the system is
\begin{equation}
    \ket{-}_C\ket{\psi} =\frac{1}{\sqrt{2}}(\ket{0}_C\ket{\psi} - \ket{1}_C\ket{\psi}) \, .
\end{equation}
Next, conditioned on the control qubit, act on the system with the operators $AB$ if in the state $\ket{0}_C$ or with operators $BA$ is the control is in the state $\ket{1}$. Performing this sequence of controls leads to the overall state
\begin{equation}
    \frac{1}{\sqrt{2}}(\ket{0}_C AB\ket{\psi} - \ket{1}_C BA\ket{\psi}) \, .
\end{equation}
After a second Hadamard on the control register, this becomes
\begin{align}
    &\frac{1}{2}(\ket{0}_C AB\ket{\psi} + \ket{1}_C AB\ket{\psi} - \ket{0}_C BA\ket{\psi}) + \ket{1}_C BA\ket{\psi})  \notag \\
    &= \frac{1}{2} [\ket{0}_C\otimes(AB - BA)\ket{\psi} + \ket{1}_C\otimes(AB +BA)\ket{\psi})] \\
    &=\frac{1}{2}[\ket{0}_C[A,B]\ket{\psi} + \ket{1}_C\{A,B\}\ket{\psi}] \, ,
\end{align}
where $\{A,B\}$ denotes the anticommutator of $A$ and $B$. Given this final state, we measure the control qubit to be in the state $\ket{0}_C$ with probability
\begin{equation}\label{eq:commutator}
    P(\ket{0}\!\bra{0}_C) =\frac{1}{4} \left\Vert [A,B]\ket{\psi} \right\Vert_2^2 \, .
\end{equation}
In this manner, we have effectively reproduced Equations~(5)--(8) of~\cite{switch3}.

\begin{figure}
\begin{center}
\includegraphics[width=\columnwidth]{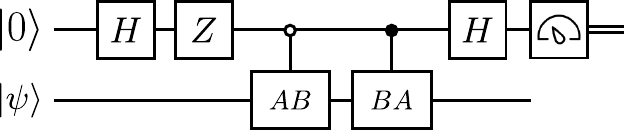}
\end{center}
\caption{Figure demonstrating how to measure the commutator of two unitary operators $A$ and $B$. The probability to measure $\ket{0}$ on the control state is $\frac{1}{4} \left\Vert [A,B]\ket{\psi}\right\Vert_2^2$}
\label{fig:commutator}
\end{figure}

However, in principle there could be higher order, or nested, commutators to consider. We can determine the sum of the nested commutators between two operators $A$ and $B$ via the Baker--Campbell--Hausdorff (BCH) formula. For some elements $X$ and $Y$ in the Lie algebra of a Lie group, BCH (specifically, an 1897 lemma of the famous theorem) tells us that
\begin{equation}\label{eq:BCH}
    e^{X}Ye^{-X} = e^{\operatorname{ad}_{X}}Y = Y + [X,Y] + \frac{1}{2!}[X,[X,Y]] + \cdots \, ,
\end{equation}
where $\operatorname{ad}_X$ is the adjoint endomorphism defined to act as $\operatorname{ad}_X Y = [X,Y]$.

\begin{figure}
\begin{center}
\includegraphics[width=\columnwidth]{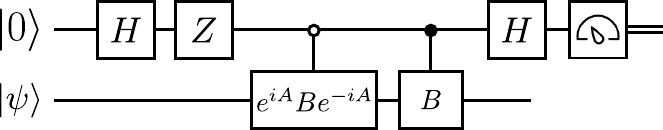}
\end{center}
\caption{Figure demonstrating how to measure the commutator of two operators $A$ and $B$. This circuit uses the BCH formula to calculate the full sum of nested commutations.}
\label{fig:BCHcommutator}
\end{figure}

Suppose now that we have an exponentiated unitary $U = e^{iA}$. In this context, $A$ is essentially a Hamiltonian, and the unitary $U$ can be realized via Hamiltonian simulation methods \cite{childs2018simulation}. For particular choices of the Hamiltonian $A$, the unitary $U$ can be efficiently realized on a quantum computer. 
While every unitary can be written in this form via Stone's theorem~\cite{hall2013quantum}, for practical purposes, it is best to choose $A$ such that $U$ can be implemented efficiently on a quantum computer. We can use~\eqref{eq:BCH} to calculate the sum of the nested commutators. Observe that
\begin{equation}
\label{eq:BCHminusB}
\begin{aligned}
    e^{iA}Be^{-iA} - B & = i[A,B] + \frac{-1}{2!}[A,[A,B]] + \cdots \\
    & = \sum_{n=1}^{\infty} (\operatorname{ad}_{iA})^n B.
\end{aligned}
\end{equation}
Realizing this quantity merely employs a rewriting of the previous algorithm used for the commutator. Replace the operator $AB$ with $e^{iA}Be^{-iA}$ and $BA$ with $B$, as shown in Figure~\ref{fig:BCHcommutator}. Then, follow the same procedure as below; only now, the final state of the system before the measurement is 
\begin{multline}
        \frac{1}{2}(\ket{0}_C(e^{iA}Be^{-iA} - B)\ket{\psi} + \ket{1}_C(e^{iA}Be^{-iA} + B)\ket{\psi}) \\
        =\frac{1}{2}\left(\sum_{n=1}^{\infty} \ket{0}_C(\operatorname{ad}_{iA})^n B\ket{\psi} + \ket{1}_C(e^{iA}Be^{-iA} + B)\ket{\psi}\right) \,.
\end{multline}
From this final state, we measure the control qubit to be in the state $\ket{0}_C$ with probability
\begin{equation}\label{eq:nestedcomm}
    P(\ket{0}\!\bra{0}_C) = \frac{1}{4}\left\Vert \sum_{n=1}^{\infty} (\operatorname{ad}_{iA})^n B \ket{\psi} \right\Vert_2^2 \, .
\end{equation}

Both the probabilities in~\eqref{eq:commutator} and~\eqref{eq:nestedcomm} above depend on the input state $\ket{\psi}$. By choosing to replace $\ket{\psi}$ with the maximally-mixed state, we can obtain an estimation of both the commutator and the sum of nested commutators respectively as
\begin{align}
    P =\frac{1}{4 d} \left\Vert [A,B] \right\Vert_2^2 \,
\end{align}
and 
\begin{equation}
    P = \frac{1}{4d}\left\Vert \sum_{n=1}^{\infty} (\operatorname{ad}_{iA})^n B  \right\Vert_2^2 \, ,
\end{equation}
where $d$ is the dimension of the input state. However, by instead maximizing all possible input states $\ket{\psi}$, we instead obtain
\begin{equation}
     P =\frac{1}{4} \left\Vert [A,B] \right\Vert_{\infty}^2 \,
\end{equation}
for the single commutator estimation and
\begin{equation}
    P = \frac{1}{4}\left\Vert \sum_{n=1}^{\infty} (\operatorname{ad}_{iA})^n B  \right\Vert_{\infty}^2 \, ,
\end{equation}
for the algorithm in which we consider the sum of the nested commutators of $iA$ and $B$. 

Thus, via this simple application of an asymmetric circuit we have demonstrated methods by which the commutator of two operators can be computed.

\subsection{Brief Comments on Schur Sampling}\label{sec:complexity}

A pointed observation of our procedure for concatenating projectors, particularly projectors onto irreps of the symmetric group as in Section~\ref{sec:ex1}, might naturally lead to the question of how this procedure compares to  Schur sampling. Schur sampling refers to the task of measuring the Schur basis labels of a $n$-qubit (or, in principle, qudit) system~\cite{Kirby2018schur,Krovi2019efficienthigh,cervero2023weak}. This task is often accomplished by first implementing the quantum Schur transform~\cite{bacon2005quantum,harrow2005applications}. The Schur basis labels are typically given in a triple $(\lambda, p_\lambda, q_\lambda)$, which denote the Young label and bases of the irreducible representations given by Schur--Weyl duality. Determining the values of each of these labels is called strong Schur sampling, while the task of only determining the Young label is called weak Schur sampling~\cite{cervero2023weak}. For more information on the quantum Schur transform and the related tasks of Schur sampling, see Refs.~\cite{bacon2005quantum,harrow2005applications}.

Much previous literature has been devoted to implementing the quantum Schur transform and measuring Young labels.
The seminal work of~\cite{bacon2005quantum} realized the quantum Schur transform in $O(\operatorname{poly}(n,d,\log (1/\varepsilon)))$ operations, with $d$ being the dimension of a qudit, $\varepsilon$ the allowable error, and $n$ is index of the symmetric group $S_n$. Similarly, more recent scholarship affirms that high-dimensional quantum Schur transforms can be implemented in $O(\operatorname{poly}(n,\log (1/\varepsilon))$ operations, and an explicit gate count of $O(n^4\log(n/\varepsilon))$ was given in~\cite{Kirby2018schur} for the inverse Schur transform. These approaches all use $O(n)$ qubits. A technique for weak Schur sampling can estimate Young labels using logarithmically many qubits and $O(n^3 \log(n/\varepsilon))$ gates~\cite{cervero2023weak}. So how would an approach based on our concatenated circuits fair in comparison?

The natural procedure for Schur sampling based on the concatenated projector procedure given in Section~\ref{sec:concat} would be to implement a projector onto each of the $r$ irreps of $S_n$ in turn and measure the projection of a state on that system. The number of qubits necessary for this task would require at least $n$ from the data register and an additional $n^2 + r$ ancillary qubits, based on the circuit construction for the symmetric group given in~\cite{laborde2021testing}, and $r-1$ qubits to store each projection measurement outcome. The reliance on ancillary qubits could in principle be reduced, as the optimal implementation based on group size would be $\log_d (|S_n|) = \log_d (n!)$, where $d$ is again the dimension of a qudit. For sufficiently large $n$, the Stirling approximation gives that this scales asymptotically as $O(n \log(n))$ with the size of the group. As such, the number of qubits required goes at best as $O(n \log(n) + n + r)$. 
For the construction used in this work, each transposition is realized by either a controlled-SWAP gate or such a gate with an additional Pauli-$Z$ acting on the ancilla. Each of these gates requires five to seven elementary gates from the Clifford+$T$ gate set to be realized~\cite{cervero2023weak} if not native. The symmetric test for $S_n$ uses $O(n^2)$ such gates. Assuming each irrep projector scales similarly, a procedure for testing each of these projections in sequence for $r$ projectors would scale as $O(5rn^2)$ elementary gates. The number of irreps of $S^n$ corresponds directly to the partitions of $n$, which itself grows exponentially with respect to $\sqrt{n}$~\cite{andrews2004integer,fulton2013representation}. As such, alternate approaches would likely be preferable to reduce necessary resources. Indeed, the authors of~\cite{bacon2005quantum} predicted just such an inefficiency, and instead achieve better results by employing Clebsch--Gordon transforms recursively.

However, our procedure does offer a benefit in an important regard. If only the projections onto the symmetric and antisymmetric spaces are desired, all other projections can be grouped together into a collective result. In this case, the algorithm described in Example~\ref{sec:ex1} will suffice. That procedure would scale primarily with the size of the symmetric and antisymmetric tests, thus requiring around $O(2 n^2)$ operations to implement.

\section{Discussion}\label{sec:discussion}

In this paper, we have provided a variety of procedures for realizing projectors in a quantum circuit, with a view towards combining symmetry, asymmetry, and antisymmetry. We have introduced an algorithm to realize the antisymmetric projector and demonstrated how an antisymmetrized circuit can be used to test for relevant physical phenomenon such as Schmidt rank, Werner invariance, or to realize a quantum switch. For the Schmidt-rank and Werner-state tests, we have given experimental results for various examples with small-dimensional systems, using currently available quantum hardware.

Some open questions remain. First and foremost, it remains to be seen whether this general procedure for concatenating projectors can be used to realize a Schur transform or a Schur sampling algorithm in the manner of~\cite{bacon2005quantum,harrow2005applications,cervero2023weak}. The naïve approach presented here would scale poorly with the size of the symmetric group. Previous literature utilized a recursive structure to simplify the implementation of this task to scale reasonably on a quantum computer. A natural next step would be to combine the recursive structure used in~\cite{bacon2005quantum} with the concatenation procedure given here to see if a less resource-intensive algorithm would result.

Another open question is whether these algorithms can be utilized to efficiently prepare the antisymmetric state. The natural approach of utilizing the maximally-mixed state as input to an antisymmetry test and then post-selecting for successful states would result in an exponentially low acceptance probability with respect to the number of qubits in said state. If such an algorithm cannot be devised to run efficiently, the natural alternative question is whether there exists an efficiently preparable state such that the acceptance probability is appreciably large.

Both of these questions rely on reasonable scaling of resources in terms of gates and qubits and are each worthy of further study.

\begin{acknowledgments}

MLL acknowledges support from the OUSD(R\&E) SMART SEED grant and NSF Grant No.~2315398.
SR and MMW acknowledge support
from the School of Electrical and Computer Engineering
at Cornell University, the National Science Foundation under
Grant No.~2315398, and AFRL under agreement no.~FA8750-23-2-0031.

This material is based on research sponsored by Air Force
Research Laboratory under agreement number FA8750-23-
2-0031. The U.S. Government is authorized to reproduce
and distribute reprints for Governmental purposes notwithstanding any copyright notation thereon. The views and
conclusions contained herein are those of the authors and
should not be interpreted as necessarily representing the official policies or endorsements, either expressed or implied,
of Air Force Research Laboratory or the U.S. Government.

This document is Distribution Statement A, Approved for Public Release; distribution is unlimited.

We acknowledge use of the IBM~Q for this work. The views expressed are those of the authors and do not reflect the official policy or position of IBM or the IBM~Q team.
\end{acknowledgments}

\bibliography{main}

\appendix

\section{Proof of Equation~\eqref{eq:resolution-of-identity-end-result}}

\label{app:proof-res-identity}

In this appendix, we examine Example~\ref{ex:res_identity} and give a derivation for the acceptance probability given in~\eqref{eq:resolution-of-identity-end-result}. Before we begin with the proof, we show some preliminary results.
\begin{lemma}
    Consider a resolution of the identity 
    \begin{equation}
        \mathbb{I} = \sum_{i=0}^{r-1} P_i,
    \end{equation}
    where $P_i^\dagger = P_i$, and $P_i^2 = P_i$. Then, 
    \begin{equation}
        P_i P_j = \delta_{i, j} P_i.
    \end{equation}
\end{lemma}

\begin{proof}
    Consider the following:
    \begin{align}
        P_j &= P_j \mathbb{I} P_j  \\
        &= P_j \left( \sum_{i=0}^{r-1} P_i \right) P_j  \\
        &= \sum_{i=0}^{r-1} P_j P_i P_j  \\
        &= P_j + \sum_{i, i \neq j} P_j P_i P_j.
    \end{align}
    Thus, 
    \begin{equation}
    \label{eq:sum_is_zero}
        \sum_{i, i \neq j} P_j P_i P_j = 0.
    \end{equation}
    Next, consider that, for every $i$ and $j$ such that $i \neq j$, and after defining $\vert \phi \rangle \coloneqq P_i P_j \vert \psi \rangle$,
    \begin{align}
        \langle \psi \vert P_j P_i P_j \vert \psi \rangle &= \langle \psi \vert P_j P_i P_i P_j \vert \psi \rangle \\
        &= \langle \phi \vert \phi \rangle  \\
        &\geq 0,
    \end{align}
    for every vector $\vert \psi \rangle$. Thus, $P_j P_i P_j$ is a positive semi-definite operator, for all $i \neq j$. Using~\eqref{eq:sum_is_zero}, a sum of positive semi-definite operators is equal to zero if and only if each terms is zero. Thus,
    \begin{equation}
        P_j P_i P_j = 0 \quad \forall i \neq j.
    \end{equation}
    Lastly, consider that, for all $i \neq j$,
    \begin{align}
        \left\Vert P_i P_j \right\Vert^2_2 &= \operatorname{Tr}[P_j P_i P_i P_j]  \\
        &= \operatorname{Tr}[P_j P_i P_j]  \\
        &= 0.
    \end{align}
    Using the fact that the two-norm of an operator is zero if and only if the operator itself is zero, we see that $P_i P_j = 0$ for $i \neq j$, completing the proof.
\end{proof}

\medskip

Next, we show that the operator $U$, defined in \eqref{eq:def-U-res-I}, is indeed unitary.
Consider that
\begin{equation}
    U = \sum_{j=0}^{r-1} \exp{\left(\frac{2\pi \mathrm{i} j}{r}\right)} P_j.
\end{equation}
Thus, we see that
\begin{align}
    UU^\dagger &= \left(\sum_{j=0}^{r-1} \exp{\left(\frac{2\pi \mathrm{i} j}{r}\right)} P_j\right) \left(\sum_{k=0}^{r-1} \exp{\left(\frac{-2\pi \mathrm{i} k}{r}\right)} P_k\right) \notag \\
    &= \sum_{j=0}^{r-1} \sum_{k=0}^{r-1} \exp{\left(\frac{2\pi \mathrm{i} (j-k)}{r}\right)} P_jP_k \notag \\
    &= \sum_{j=0}^{r-1} \sum_{k=0}^{r-1} \exp{\left(\frac{2\pi \mathrm{i} (j-k)}{r}\right)} \delta_{j, k} P_j \notag \\
    &= \sum_{j=0}^{r=1} P_j \notag \\
    &= \mathbb{I}.
\end{align}
Similarly, we can show that $U^\dagger U = \mathbb{I}$. 

The initial state of the system, in the algorithm described just before \eqref{eq:resolution-of-identity-end-result}, is given by
\begin{equation}
    \ket{000}_{C_3C_2C_1} \ket{\psi}_{S_1R_1} \ket{\psi}_{S_2R_2}.
\end{equation}
Next, we act on system $C_3$ with the Hadamard operation, and on $C_2$ and $C_1$ each with the quantum Fourier transform (QFT). The overall state is now
\begin{equation}
    \frac{1}{r} \ket{+}_{C_3} \sum_{a,b=0}^{r-1} \ket{ab}_{C_2C_1} \ket{\psi} \ket{\psi}.
\end{equation}
Next, we perform the two controlled operations
\begin{align}
    &\sum_{a=0}^{r-1} \outerproj{a}_{C_1} \otimes U^a_{A_2B_2}, \\
    &\sum_{d=0}^{r-1} \outerproj{b}_{C_2} \otimes U^b_{A_1B_1}.
\end{align}
The overall state is then
\begin{equation}
    \frac{1}{r} \ket{+} \sum_{a,b=0}^{r-1} \ket{a} \ket{b} U^a\ket{\psi} \otimes U^b\ket{\psi}.
\end{equation}
Next, we perform the $\tilde{Z}$ operation on each of the systems $C_2$ and $C_1$, where the operation is defined by
\begin{equation}
    \tilde{Z} \coloneqq \sum_{c=0}^{r-1} \exp{\left(\frac{2\pi \mathrm{i} j}{r}\right)} \outerproj{c}.
\end{equation}
The state is then given by
\begin{equation}
    \frac{1}{r} \ket{+} \sum_{a,b=0}^{r-1} \exp{\left(\frac{2\pi \mathrm{i} (a+b)}{r}\right)} \ket{a} \ket{b} U^a\ket{\psi} \otimes U^b\ket{\psi}.
\end{equation}
\begin{widetext}
Performing the inverse QFT (IQFT) on each of the systems $C_1$ and $C_2$ next leads to the following state:
\begin{align}
    &\frac{1}{r^2} \ket{+} \sum_{a,b=0}^{r-1} \sum_{c,d=0}^{r-1} \exp{\left(\frac{2\pi \mathrm{i} (a+b)}{r}\right)} \exp{\left(\frac{-2\pi \mathrm{i} (ac+bd)}{r}\right)}
    \ket{c} \ket{d} U^a\ket{\psi} \otimes U^b\ket{\psi} \notag \\
    &= \frac{1}{r^2} \ket{+} \sum_{c,d=0}^{r-1} \ket{c} \ket{d}  \Bigg[ \sum_{a,b=0}^{r-1} \exp{\left(\frac{2\pi \mathrm{i} (a+b)}{r}\right)} \exp{\left(\frac{-2\pi \mathrm{i} (ac+bd)}{r}\right)} U^a \otimes U^b \Bigg] \ket{\psi} \ket{\psi} \notag \\
    &= \frac{1}{r^2} \ket{+} \sum_{c,d=0}^{r-1} \ket{c} \ket{d} \Bigg[ \sum_{a=0}^{r-1} \exp{\left(\frac{2\pi \mathrm{i} a}{r}\right)} \exp{\left(\frac{-2\pi \mathrm{i} ac}{r}\right)} U^a \Bigg] \ket{\psi} \otimes \Bigg[ \sum_{b=0}^{r-1} \exp{\left(\frac{2\pi \mathrm{i} b}{r}\right)} \exp{\left(\frac{-2\pi \mathrm{i} bd}{r}\right)} U^b \Bigg] \ket{\psi}.
\end{align}
Consider the following:
\begin{align}
     \sum_{a=0}^{r-1} \exp{\left(\frac{2\pi \mathrm{i} a}{r}\right)} \exp{\left(\frac{-2\pi \mathrm{i} ac}{r}\right)} U^a 
    &= \sum_{a=0}^{r-1} \sum_{j=0}^{r-1} \exp{\left(\frac{2\pi \mathrm{i} a}{r}\right)} \exp{\left(\frac{-2\pi \mathrm{i} ac}{r}\right)} \exp{\left(\frac{2\pi \mathrm{i} ja}{r}\right)} P_j \notag \\
    &= \sum_{j=0}^{r-1} r \delta_{j+1,c} P_j \notag \\
    &= r P_{c-1},
\end{align}
where the subscript $c-1$ is taken modulo $r$. Thus, the overall state after the IQFT  is given by
\begin{equation}
    \ket{+} \sum_{c,d=0}^{r-1} \ket{c} \ket{d} P_{c-1} \ket{\psi} 
    \otimes P_{d-1} \ket{\psi}.
\end{equation}
Next, we apply the three CSWAP operations and the Hadamard operation
\begin{equation}
    \frac{1}{\sqrt{2}} \Bigg[ \ket{+} \sum_{c,d=0}^{r-1} \ket{c} \ket{d} P_{c-1} \ket{\psi} \otimes P_{d-1} \ket{\psi} + \ket{-} \sum_{c,d=0}^{r-1} \ket{c} \ket{d} F(P_{d-1} \ket{\psi} \otimes P_{c-1} \ket{\psi}) \Bigg].
\end{equation}
Expanding, we get
\begin{equation}
\label{eq:app-Phi-state}
    \ket{\Phi} = \frac{1}{2} \sum_{y=0}^{1} \sum_{c,d=0}^{r-1} \left[ \ket{y}\ket{cd} (P_{c-1} \otimes P_{d-1} + (-1)^y F(P_{d-1} \otimes P_{c-1})) \right].
\end{equation}

\end{widetext}
Lastly, measuring registers $C_3 C_2 C_1$, we find that
\begin{equation}
\begin{aligned}
   &  p(xab)  \\
   & = \left\Vert \bra{xab}_{C_3 C_2 C_1} \otimes I_{S_1 R_1 S_2 R_2} \ket{\Phi} \right\Vert^2_2.  \\
   &= \frac{1}{2} [\Tr{P_{a-1} \rho}\Tr{P_{b-1} \rho} + (-1)^x \Tr{P_{a-1} \rho P_{b-1} \rho}].
\end{aligned}
\end{equation}
Thus, it follows that
\begin{equation}
    p(0ab) - p(1ab) = \Tr{P_{a-1}\rho P_{b-1}\rho}.
    \label{eq:p0ab-p1ab}
\end{equation}
Observe that 
\begin{equation}
\begin{aligned}
    p(0ba) - p(1ba) & = \Tr{P_{b-1}\rho P_{a-1}\rho}\\
     & = \Tr{P_{a-1}\rho P_{b-1}\rho },
\end{aligned}
    \label{eq:p0ba-p1ba}
\end{equation}

As such, to estimate this quantity, we can perform a sampling procedure similar to what is described around \eqref{eq:mu-mean}. For all $t\in \{1, \ldots, T\}$, let $Y_t$ be a random variable that is set to  $+1$ after obtaining the results $0ab$ or $0ba$, $-1$ after obtaining the results $1ab$ or $1ba$, and 0 otherwise. Then, it follows from~\eqref{eq:p0ab-p1ab}--\eqref{eq:p0ba-p1ba} and \eqref{eq:cross-term} that
\begin{align}
\mu & \equiv \mathbb{E}[Y_t] \\
& = 2 \Tr{P_{a-1}\rho P_{b-1}\rho } \\
& = 4 \left \Vert \operatorname{Re}(P_{a-1} \rho P_{b-1})\right \Vert_2^2,
\end{align}
so that $\overline{Y_T} \coloneqq \frac{1}{T}\sum_{t=1}^T Y_t$ is an unbiased estimator of $4 \left \Vert \operatorname{Re}(P_{a-1} \rho P_{b-1})\right \Vert_2^2$. Set $\varepsilon > 0$ and $\delta \in (0,1)$. The Hoeffding bound (see, e.g., \cite[Lemma~1]{rethinasamy2023quantum})  implies that, if $T \geq \frac{2}{\varepsilon^2} \ln \!\left( \frac{2}{\delta}\right)$, then $\Pr[\vert \overline{Y_T} - \mu \vert \leq \varepsilon] \geq 1-\delta$.
\end{document}